\documentclass[a4paper,11pt]{article}
\usepackage[a4paper, margin=1in]{geometry}
\usepackage[utf8]{inputenc}
\usepackage[T1]{fontenc}
\usepackage{graphicx}
\usepackage{amsmath, amssymb, amstext, mathtools}
\usepackage{amsthm}
\usepackage{tikz}
\usepackage{pgfplots}
\usepackage{hyperref}
\usepackage[titlenumbered, linesnumbered, ruled]{algorithm2e}

\usepackage{enumitem}
\setenumerate[1]{label=(\arabic*)}

\usetikzlibrary{patterns,arrows,decorations.pathreplacing,decorations.pathmorphing,calc}

\hypersetup{
 breaklinks=true,
 colorlinks=true,
 citecolor=purple!70!blue!60!black,
 linkcolor=purple!70!blue!90!black,
 urlcolor=black,
 pdflang={en},
 pdftitle={Canonisation and Definability for Graphs of Bounded Rank Width},
 pdfauthor={Martin Grohe and Daniel Neuen}
}

\newtheorem{theorem}{Theorem}[section]
\newtheorem{lemma}[theorem]{Lemma}
\newtheorem{corollary}[theorem]{Corollary}
\newtheorem{example}[theorem]{Example}
\newtheorem{observation}[theorem]{Observation}
\theoremstyle{definition}
\newtheorem{definition}[theorem]{Definition}
\theoremstyle{remark}
\newtheorem{remark}[theorem]{Remark}
\newtheorem{claim}[theorem]{Claim}

\newenvironment{claimproof}{\begin{proof}}{\end{proof}}

\newcommand{\case}[1]{\par\medskip\noindent\textit{Case #1: }}

\newenvironment{cs}{
  \begin{description}
    \renewcommand{\case}[1]{\item[\itshape\mdseries Case ##1:]}
  }{
  \end{description}
}

\DeclareMathOperator{\argmin}{argmin}

\DeclareMathOperator{\rk}{rk}

\DeclareMathOperator{\xvec}{vec}
\DeclareMathOperator{\comp}{Comp}

\DeclareMathOperator{\width}{wd}
\DeclareMathOperator{\rw}{rw}
\DeclareMathOperator{\tw}{tw}

\DeclareMathOperator{\lab}{lab}
\DeclareMathOperator{\cw}{cw}

\DeclareMathOperator{\BP}{BP}

\newcommand{\logic}[1]{\textsf{\upshape\small #1}}
\newcommand{\FO}{\logic{FO}}
\newcommand{\LC}{\logic C}
\newcommand{\LFP}{\logic{LFP}}
\newcommand{\IFP}{\logic{IFP}}
\newcommand{\FPC}{\logic{FP+C}}

\newcommand{\formel}[1]{\textsf{\upshape\small #1}}

\newcommand{\LL}{\logic L}
\newcommand{\ifp}{\operatorname{ifp}}

\newcommand{\PTIME}{\logic{PTIME}}

\renewcommand{\phi}{\varphi}
\newcommand{\bigmid}{\;\big|\;}
\newcommand{\Bigmid}{\;\Big|\;}

\newcommand{\CP}{\mathcal P}
\newcommand{\CC}{\mathcal C}

\newcommand{\anc}{{\sf anc}}

\newcommand{\notleftright}{\mathrel{\ooalign{$\Leftrightarrow$\cr\hidewidth$/$\hidewidth}}}

\newcommand{\orcid}[1]{\href{https://orcid.org/#1}{\includegraphics[height=1.8ex]{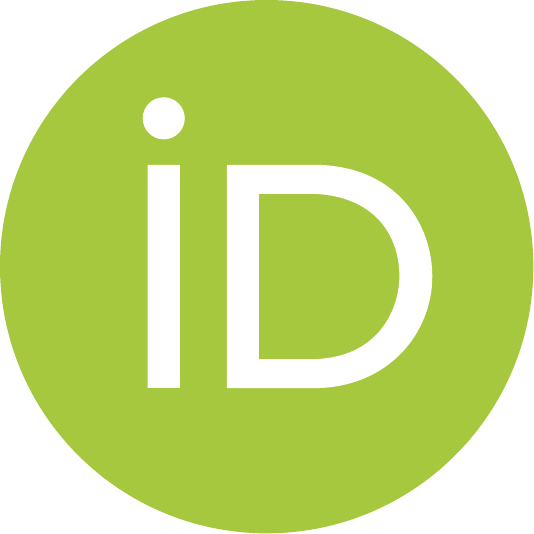}}}
\newcommand{\email}[1]{\href{mailto:#1}{\texttt{#1}}}

\title{Canonisation and Definability\\for Graphs of Bounded Rank Width}
\author{
Martin Grohe \orcid{0000-0002-0292-9142}\\
RWTH Aachen University\\
\email{grohe@informatik.rwth-aachen.de}
\and
Daniel Neuen \orcid{0000-0002-4940-0318}\\
Simon Fraser University\\
\email{dneuen@sfu.ca}
}
\date{}

\begin{document}

\maketitle

\begin{abstract}
 We prove that the combinatorial Weisfeiler-Leman algorithm of
 dimension $(3k+4)$ is a complete isomorphism test for the class of
 all graphs of rank width at most $k$. Rank width is a graph
 invariant that, similarly to tree width, measures the width of a
 certain style of hierarchical decomposition of graphs; it is
 equivalent to clique width.

 It was known that isomorphism of graphs of rank width $k$ is
 decidable in polynomial time (Grohe and Schweitzer, FOCS 2015), but the best
 previously known algorithm has a running time $n^{f(k)}$ for a
 non-elementary function $f$. Our result yields an isomorphism test
 for graphs of rank width $k$ running in time $n^{O(k)}$. Another
 consequence of our result is the first polynomial-time canonisation algorithm
 for graphs of bounded rank width.

 Our second main result is that fixed-point logic with counting
 captures polynomial time on all graph classes of bounded rank width.
\end{abstract}

\section{Introduction}

Rank width, introduced by Oum and Seymour~\cite{oum05,OumS06}, is a graph
invariant that measures how well a graph can be decomposed
hierarchically in a certain style. In this respect, it is similar to the better-known
tree width, but where tree width measures the complexity, or width, of
a separation in such a hierarchical decomposition in terms of the
``connectivity'' between the two sides, rank width measures the
complexity of a separation in terms of the rank of the adjacency
matrix of the edges between the two sides of the separation. This
makes rank width (almost) invariant under complementation of a graph
and thus relevant for dense graphs, where tree width usually becomes
meaningless. Rank width is closely related to clique width, which had
been introduced
by Courcelle and Olariu~\cite{couola00}: for every graph $G$ it holds
that $\rw(G)\le\cw(G)\le2^{\rw(G)+1}-1$, where $\rw(G)$ denotes the
rank width and $\cw(G)$ the clique width of $G$. This implies that
many hard algorithmic problems can be solved efficiently on graphs of
bounded rank width (see, for example,~\cite{espegurwan01}), among them all problems definable in
monadic second-order logic~\cite{coumakrot00}. Furthermore, graph
classes of bounded clique width, or equivalently bounded rank width,
are precisely those that can be obtained by means of a monadic
second-order transduction over a class of trees
\cite{cou95a,coueng95,couengroz93}.

In this paper we study the graph isomorphism problem and the
closely related graph canonisation problem as well as logical
definability and descriptive complexity on graph classes of bounded
rank width. 

Despite Babai's quasipolynomial time algorithm
\cite{Babai16}, it is still wide open whether the graph isomorphism
problem can be solved in polynomial time. Polynomial-time algorithms
are only known for specific graph classes, among them all classes of
bounded degree~\cite{luk82}, all classes of bounded tree
width~\cite{bod90,LokshtanovPPS17}, all classes excluding a fixed graph as a
minor~\cite{pon88}, even all classes excluding a fixed graph as a
topological subgraph~\cite{gromar15}, and most recently, graph classes
of bounded rank width \cite{GroheS15}. This last result
was the starting point for our present paper. The running time of
the isomorphism test in \cite{GroheS15} is $n^{f(k)}$, where $n$ is
the number of vertices and $k$ the rank width of the input graph, and
$f$ is a \emph{non-elementary} function. Of course this is
unsatisfactory. Moreover, the algorithm is extremely complicated,
using both advanced techniques from structural graph theory
\cite{gm10,OumS06,GroheS16} and the group-theoretic graph isomorphism
machinery~\cite{luk82}.

Our first contribution is a simple isomorphism test for graphs of rank
width at most $k$ running in time $n^{O(k)}$. Indeed, the algorithm we
use is a generic combinatorial isomorphism test known as the
Weisfeiler-Leman algorithm \cite{weilem68,Babai16,caifurimm92}. The \emph{$\ell$-dimensional
Weisfeiler Leman} algorithm ($\ell$-WL) iteratively colours
$\ell$-tuples of vertices of the two input graphs and then compares
the resulting colour patterns. If they differ, we know that the two
input graphs are nonisomorphic. If two graphs have the same colour
pattern, in general they may still be nonisomorphic
\cite{caifurimm92}. Thus, $\ell$-WL is not a \emph{complete}
isomorphism test for all graphs. However, we prove that it is for
graphs of bounded rank width. We say that $\ell$-WL \emph{identifies}
a graph $G$ if it distinguishes $G$ from every graph $H$ not
isomorphic to $G$.

\begin{theorem}\label{theo:1}
 The $(3k+4)$-dimensional Weisfeiler-Leman algorithm identifies every
 graph of rank width at most $k$.
\end{theorem}

Combining this theorem with a result due to Immerman and Lander on the
running time of the WL algorithm, we obtain the following.

\begin{corollary}
 Isomorphism of graphs of rank width $k$ can be decided in time
 $O(n^{3k+5}\log n)$.
\end{corollary}

Another way of stating Theorem~\ref{theo:1} is that the \emph{Weisfeiler-Leman (WL) dimension} \cite{gro17} of
graphs of rank width $k$ is at most $3k+4$. While it is known that
many natural graph classes have bounded WL dimension, among them
the class of planar graphs~\cite{gro98a,KieferPS19}, classes of bounded
genus~\cite{gro00,GroheK19}, bounded tree
width \cite{gromar99,KieferN22}, classes of graphs excluding some
fixed graph as a minor~\cite{gro17}, and interval
graphs~\cite{EvdokimovPT00}, all these except for the class of
interval graphs are classes of
sparse graphs (with an edge number linear in the number of
vertices). Our result adds a rich family of classes that include dense
graphs to the picture. 

Immerman and Lander~\cite{immlan90} (also see \cite{caifurimm92}) showed that
$\ell$-WL is an equivalence test for $\LC^{\ell+1}$, the $(\ell+1)$-variable fragment of first-order logic
with counting. Hence our result can also be read as a definability result.

\begin{corollary}\label{cor:definability}
 For every graph $G$ of rank width at most $k$ there is a sentence
 $\phi_G$ of the logic $\LC^{3k+5}$ that characterises $G$ up to
 isomorphism.
\end{corollary}
 
We use this connection to logic in our proof of Theorem~\ref{theo:1}, which is
based on a characterisation of equivalence in the logic $\LC^\ell$ in
terms of an Ehrenfeucht-Fra\"{\i}ss\'e game, the so-called
\emph{$\ell$-bijective pebble game} due to Hella~\cite{hel96}.

A \emph{canonisation algorithm} $A$ for a class $\CC$ of graphs associates with each
graph $G\in\CC$ a graph $A(G)$ that is isomorphic to $G$ in such a way
that if $G,H\in\CC$ are isomorphic then $A(G)$ and $A(H)$ are
identical. Clearly, a canonisation algorithm can be used to test if
two graphs are isomorphic; the converse is not known. 
It is known,\footnote{The result is certainly not new and not ours, but
unfortunately we are not aware of a reference. We sketch a proof in
Appendix \ref{sec:app-canon-from-wl}.}
however, that if a class of graphs has WL
dimension at most $\ell$ then there is a canonisation algorithm for
this class running in time $O(n^{\ell+3}\log n)$. Hence, as another corollary
to Theorem~\ref{theo:1}, we obtain the first polynomial-time
canonisation algorithm for graphs of bounded rank width.

\begin{corollary}\label{cor:canonisation}
 There is a canonisation algorithm for the class of graphs of rank
 width at most $k$
 running in time $O(n^{3k+7}\log n)$.
\end{corollary}

The second part of our paper is concerned with descriptive complexity
theory. The central open question of the field is whether there is a
logic that \emph{captures polynomial time}
\cite{chahar82,gur88}. Intuitively, this means that all sentences of
the logic can be evaluated in polynomial time (by a uniform algorithm)
and that all polynomial-time decidable properties can be defined in
the logic. By the Immerman-Vardi Theorem \cite{imm87,var82},
least-fixed point logic $\LFP$ captures polynomial time on ordered
structures (that is, structures with a distinguished binary relation
that is a linear order of the universe). But for general structures
the question is still wide open more than 35 years after it was first
raised by Chandra and Harel~\cite{chahar82}.  The question is related
to isomorphism testing and canonisation. Indeed, a polynomial-time
canonisation algorithm for the class of all graphs would imply that
there is a logic capturing polynomial time.

The question for a logic capturing polynomial time, as formulated by
Gurevich~\cite{gur88}, casts the notion of what constitutes a logic
deliberately wide. However, we are not mainly interested in an
abstract logic, but in a ``nice'' logic that conveys some insights on
the nature of polynomial-time computation. A logic that arguably falls
in this category is \emph{fixed-point logic with counting $\FPC$},
first proposed by Immerman~\cite{imm87} and later formalised by
Gr\"adel and Otto~\cite{graott93} in the form commonly used today. It
is known that $\FPC$ does not capture polynomial time
\cite{caifurimm92}. But over the last 10 years it has become clear
that the logic is surprisingly powerful. It captures specific polynomial
time algorithms such as linear programming \cite{anddawhol15}, and it
does capture polynomial time on rich graph classes, including all
classes excluding some fixed graph as a minor~\cite{gro17}. Our second main result
further broadens the scope of $\FPC$-definability.

\begin{theorem}\label{theo:2}
 For every $k$, fixed-point logic with counting $\FPC$ captures
 polynomial time on the class of all graphs of rank width at most $k$.
\end{theorem}

Technically, this theorem is related to the first and is
based on the same graph-theoretic ideas, but it is significantly harder
to prove. On an abstract level, this can be explained by highlighting
an important difference between Theorem~\ref{theo:2} and Corollary~\ref{cor:definability},
which rephrases Theorem~\ref{theo:1} in terms of
logic. Corollary~\ref{cor:definability} is a \emph{nonuniform}
definability result: for every fixed graph we construct a formula
characterising this graph. By contrast, Theorem~\ref{theo:2} requires
\emph{uniform} definability: for every polynomial time property we
want a unique sentence that defines this property for all graphs of
rank width at most $k$. This means that we have to internalise the
construction that is underlying the proof of Theorem~\ref{theo:1} in
the logic \FPC.

The paper is organised as follows: after reviewing the necessary
preliminaries on rank width, graph isomorphism testing, and the WL
algorithm in Section~\ref{sec:prels}, in Section~\ref{sec:split} we introduce our
technical machinery for dealing with rank decompositions that is
underlying the proofs of both theorems. We prove Theorem~\ref{theo:1}
in Section~\ref{sec:WLdim} and Theorem~\ref{theo:2} in
Section~\ref{sec:cap}, after giving additional background in descriptive
complexity theory in Subsection~\ref{sec:descriptive}.

\section{Preliminaries}
\label{sec:prels}

\subsection{Graphs}

A \emph{graph} is a pair $G=(V,E)$ with vertex set $V = V(G)$ and edge relation $E = E(G)$.
In this paper all graphs are finite, simple (no loops or multiple
edges), and undirected. We denote edges by $vw \in E(G)$ where $v,w \in V(G)$.
The \emph{neighbourhood} of~$v\in V(G)$ is denoted by~$N(v)$.
For $A \subseteq V(G)$ we denote by $G[A]$ the \emph{induced subgraph} of $G$ on $A$.
Also, we denote by $G\setminus A$ the induced subgraph on the complement of $A$, that is $G \setminus A \coloneqq G[V(G) \setminus A]$.

An \emph{isomorphism} from a graph $G$ to another graph $H$ is a bijective mapping $\varphi\colon V(G) \rightarrow V(H)$ which preserves the edge relation, that is, $vw \in E(G)$ if and only if $\varphi(v)\varphi(w) \in E(H)$ for all~$v,w \in V(G)$.
Two graphs $G$ and $H$ are \emph{isomorphic} ($G \cong H$) if there is an isomorphism from~$G$ to~$H$.
We write $\varphi\colon G\cong H$ to denote that $\varphi$ is an isomorphism from $G$ to $H$.

A (vertex-)coloured graph is a tuple $(G,\chi)$ where $\chi\colon V(G) \rightarrow \mathcal{C}$ is a mapping and $\mathcal{C}$ is a finite set of colours.
Typically the set of colours is just an initial segment $[n] \coloneqq \{1,\dots,n\}$ of the natural numbers.
Isomorphisms between coloured graphs have to respect the colours of the vertices.
In this paper, we typically consider coloured graphs also when not explicitly stated.
Note that an uncoloured graph may be viewed as a coloured graph where each vertex gets the same colour.

\subsection{Rank Width and Clique Width}

In this work, we are interested in graphs of bounded rank width and graphs of bounded clique width.
This section formally defines both parameters and describes the basic connections between them.

\paragraph{Rank Width}

Rank width is a graph invariant that was first introduced by Oum and Seymour \cite{OumS06} and which measures the width of a certain style of hierarchical decomposition of graphs.
Intuitively, the aim is to repeatedly split the vertex set of the graph along cuts of low complexity in a hierarchical fashion.
For rank width, the complexity of a cut is measured in terms of the
rank of the matrix capturing the adjacencies between the two sides of
the cut over the 2-element field $\mathbb F_2$.

Let $G$ be a graph. For $X,Y \subseteq V(G)$ we define $M(X,Y) \in \mathbb{F}_2^{X \times Y}$ where $(M(X,Y))_{x,y} = 1$ if and only if $xy \in E(G)$.
Furthermore $\rho_G(X) \coloneqq \rk_2(M(X,\overline{X}))$ where $\overline{X} \coloneqq V(G) \setminus X$ and $\rk_2(A)$ denotes the $\mathbb{F}_2$-rank of a matrix $A$.

A \emph{rank decomposition} of $G$ is a tuple $(T,\gamma)$ consisting of a binary rooted tree $T$ and a mapping $\gamma\colon V(T) \rightarrow 2^{V(G)}$ such that
\begin{enumerate}[label = (R.\arabic*)]
 \item $\gamma(r) = V(G)$ where $r$ is the root of $T$,
 \item\label{item:rank-2} $\gamma(t) = \gamma(s_1) \cup \gamma(s_2)$ and $\gamma(s_1) \cap \gamma(s_2) = \emptyset$ for all internal nodes $t \in V(T)$ with children $s_1$ and $s_2$, and
 \item $|\gamma(t)| = 1$ for all $t \in L(T)$, where $L(T)$ denotes the set of leaves of the tree $T$.
\end{enumerate}
Note that, instead of giving $\gamma$, we can equivalently specify a bijection $f\colon L(T) \rightarrow V(G)$ (this completely specifies $\gamma$ by Condition \ref{item:rank-2}).
The \emph{width} of a rank decomposition $(T,\gamma)$ is
\[\width(T,\gamma) \coloneqq \max\{\rho_G(\gamma(t)) \mid t \in V(T)\}.\]
The \emph{rank width} of a graph $G$ is
\[
\rw(G) \coloneqq \min\{ \width(T,\gamma) \mid (T,\gamma) \text{ is a rank decomposition of } G\}.
\]

\paragraph{Clique Width}

Clique width \cite{couola00} is another measure aiming to describe the structural complexity of a graph, but unlike rank width, it considers the complexity of an algebraic expression defining the graph.

For $k \in \mathbb{N}$ a \emph{$k$-graph} is a pair $(G,\lab)$ where $G$ is a graph and $\lab\colon V(G) \rightarrow [k]$ is a labelling of vertices.
We define the following four operations for $k$-graphs:
\begin{enumerate}
 \item for $i \in [k]$ let $\cdot_i$ denote an isolated vertex with label $i$,
 \item for $i,j \in [k]$ with $i \neq j$ we define $\eta_{i,j}(G,\lab) = (G',\lab)$ where
  $V(G') \coloneqq V(G)$ and $E(G') \coloneqq E(G) \cup \{vw \mid \lab(v) = i \wedge \lab(w) = j\}$,
 \item for $i,j \in [k]$ we define $\rho_{i\rightarrow j}(G,\lab) = (G,\lab')$ where
  \[\lab'(v) \coloneqq \begin{cases}
                j &\text{if } \lab(v) = i\\
                \lab(v) &\text{otherwise}
               \end{cases},\]
 \item for two $k$-graphs $(G,\lab)$ and $(G',\lab')$ we define $(G,\lab) \oplus (G',\lab')$ to be the disjoint union of the two $k$-graphs.
\end{enumerate}
A \emph{$k$-expression} $t$ is a well-formed expression in these symbols and defines a $k$-graph $(G,\lab)$.
In this case $t$ is a $k$-expression for $G$.
The \emph{clique width} of a graph $G$, denoted by $\cw(G)$, is the minimum $k \in \mathbb{N}$ such that there is a $k$-expression for $G$.

Comparing clique width and rank width, each parameter is bounded in terms of the other.

\begin{theorem}[\cite{OumS06}]
 \label{thm:bound-rw-cw}
 For every graph $G$ it holds that
 \[\rw(G) \leq \cw(G) \leq 2^{\rw(G) + 1} - 1.\]
\end{theorem}

Also, there is the following connection to tree width.

\begin{theorem}[\cite{Oum08}]
 \label{thm:bound-rw-tw}
 For every graph $G$ it holds that
 \[\rw(G) \leq \tw(G) + 1.\]
 where $\tw(G)$ denotes the tree width of $G$.
\end{theorem}

Note that the tree width of a graph can not be bounded in terms of its rank width.
For example, the complete graph on $n$ vertices $K_n$ has rank width $\rw(K_n) = 1$ and tree width $\tw(K_n) = n-1$.

\subsection{The Weisfeiler-Leman Algorithm}

The $k$-dimensional Weisfeiler-Leman algorithm is a procedure that, given a graph~$G$ and a colouring of the~$k$-tuples of the vertices, computes an isomorphism-invariant refinement of the colouring.
Let~$\chi_1,\chi_2\colon V^k \rightarrow \mathcal{C}$ be colourings of the~$k$-tuples of vertices of~$G$, where~$\mathcal{C}$ is some finite set of colours. 
We say $\chi_1$ \emph{refines} $\chi_2$ ($\chi_1 \preceq \chi_2$) if for all $\bar{v},\bar{w} \in V^k$ we have \[\chi_1(\bar{v}) = \chi_1(\bar{w}) \;\Rightarrow\; \chi_2(\bar{v}) = \chi_2(\bar{w}).\]

For an integer~$k >1 $ and a vertex-coloured graph~$(G,\chi)$, we first set $\chi_0^{G,k}\colon V^{k} \rightarrow \mathcal{C}$ to be the colouring where each $k$-tuple is coloured by the isomorphism-type of its underlying ordered subgraph.
More precisely, $\chi_0^{G,k}(v_1,\dots,v_k) = \chi_0^{G,k}(w_1,\dots,w_k)$ if and only if
for all $i\in [k]$ it holds that $\chi(v_i)= \chi(w_i)$ and for all $i,j\in [k]$  it holds $v_i = v_j \Leftrightarrow w_i =w_j$ and $v_iv_j \in E(G) \Leftrightarrow w_iw_j \in E(G)$.
Then, we recursively define the colouring~$\chi^{G,k}_{i+1}$ by setting~$\chi^{G,k}_{i+1}(v_1, \dots, v_k) \coloneqq (\chi^{G,k}_{i} (v_1, \dots, v_k); \mathcal{M})$,
where~$\mathcal{M}$ is a multiset defined as
\[\big\{\!\!\big\{\big(\chi^{G,k}_{i}(\bar v[w/1]),\chi^{G,k}_{i}(\bar v[w/2]), \dots, \chi^{G,k}_{i}(\bar v[w/k])\big) \mid w\in V \big\}\!\!\big\}\]
where $\bar v[w/i] \coloneqq (v_1,\dots,v_{i-1},w,v_{i+1},\dots,v_k)$.

For~$k=1$ the definition is similar but we iterate only over the neighbours of $v_1$, that is the multiset is defined by ${\mathcal{M}}\coloneqq  \{\!\! \{ \chi^{G,1}_i(w) \mid w\in N(v_1) \}\!\!\}$.
The initial colouring $\chi_0^{G,1}$ is simply equal to $\chi$, the vertex-colouring of the input graph.

By definition, every colouring~$\chi^{G,k}_{i+1}$ induces a refinement of the partition of the~$k$-tuples of the graph~$G$ with colouring~$\chi^{G,k}_{i}$.
Thus, there is some minimal~$i$ such that the partition induced by the colouring~$\chi^{G,k}_{i+1}$ is not strictly finer than the one induced by the colouring~$\chi^{G,k}_i$ on~$G$.
For this minimal~$i$, we call the colouring~$\chi^{G,k}_i$ the
\emph{stable} colouring of~$G$ and denote it
by~$\chi^{G,k}_{(\infty)}$.

For $k=1$ we will usually omit the index $k$ and write $\chi^{G}_{(\infty)}$ instead of $\chi^{G,k}_{(\infty)}$.
Also, in some cases we will omit the graph $G$ if it is apparent from context and just write $\chi_{(\infty)}$.

For~$k \in \mathbb{N}$, the \emph{$k$-dimensional Weisfeiler-Leman algorithm} takes as input a coloured graph $(G, \chi)$ and returns the coloured graph~$(G, \chi^{G,k}_{(\infty)})$.
This can be implemented in time $O(n^{k+1}\log n)$ \cite{immlan90}.
For two graphs~$G$ and~$H$, we say that the~$k$-dimensional Weisfeiler-Leman algorithm \emph{distinguishes}~$G$ and~$H$ if there is some colour~$c$ such that the sets
$\{\bar{v} \mid \bar{v} \in V^k(G), \chi^{G,k}_{(\infty)}(\bar{v}) = c\}$ and $\{\bar{w} \mid \bar{w} \in V^k(H), \chi^{H,k}_{(\infty)}(\bar{w}) = c\}$
have different cardinalities. We write~$G \simeq_k H$ if the $k$-dimensional Weisfeiler-Leman algorithm does not distinguish between~$G$ and~$H$.
The $k$-dimensional Weisfeiler-Leman algorithm \emph{identifies} a graph $G$ if it distinguishes $G$ from every non-isomorphic graph $H$.

\paragraph{Pebble Games}
We will not require details about the information computed by the Weisfeiler-Leman algorithm and rather use the following pebble game that is known to capture the same information.
Let $k \in \mathbb{N}$.
For graphs $G,H$ on the same number of vertices and with vertex colourings~$\chi_G$ and~$\chi_H$, respectively,
we define the \emph{bijective $k$-pebble game} $\BP_{k}(G,H)$ as follows:
\begin{itemize}
 \item The game has two players called Spoiler and Duplicator.
 \item The game proceeds in rounds. Each round is associated with a pair of positions
 $(\bar v,\bar w)$ with~$\bar v \in V(G)^\ell$ and~$\bar w \in V(H)^\ell$ where $0 \leq \ell \leq k$.
 \item The initial position of the game is $((),())$ (the pair of empty tuples).
 \item Each round consists of the following steps. Suppose the current position of the game is $(\bar v,\bar w) = ((v_1,\ldots,v_\ell),(w_1,\ldots,w_\ell))$.
  First, Spoiler chooses whether to remove a pair of pebbles or to play a new pair of pebbles.
  The first option is only possible if $\ell > 0$ and the latter option is only possible if $\ell < k$.
  
  If Spoiler wishes to remove a pair of pebbles he picks some $i \in [\ell]$ and the game moves to position
  $(\bar v\setminus i,\bar w\setminus i)$ where $\bar v \setminus i \coloneqq (v_1,\dots,v_{i-1},v_{i+1},\dots,v_\ell)$ ($\bar w \setminus i$ is defined in the same way).
  Otherwise the following steps are performed.
  \begin{itemize}
   \item[(D)] Duplicator picks a bijection $f\colon V(G) \rightarrow V(H)$.
   \item[(S)] Spoiler chooses $v \in V(G)$ and sets $w \coloneqq f(v)$.
  \end{itemize}
  The new position is then $((v_1,\dots,v_\ell,v),(w_1,\dots,w_\ell,w))$.
  
  Spoiler wins the play if for the current position~$((v_1,\dots,v_\ell),(w_1,\dots,w_\ell))$ the induced graphs are not isomorphic.
  More precisely, Spoiler wins if there is an~$i\in [\ell]$ such that $\chi_G(v_i) \neq \chi_H(w_i)$ or there are~$i,j\in [\ell]$ such that~$v_i = v_j\notleftright w_i =w_j$ or~$v_iv_j \in E(G)\notleftright w_iw_j \in E(H)$.
  If the play never ends Duplicator wins.
\end{itemize}

We say that Spoiler (resp.\ Duplicator) wins the bijective $k$-pebble game $\BP_k(G,H)$ if Spoiler (resp.\ Duplicator) has a winning strategy for the game.

\begin{theorem}[\cite{caifurimm92,hel96}]
 \label{thm:eq-wl-pebble}
 Let $G, H$ be two graphs.
 Then $G \simeq_{k} H$ if and only if Duplicator wins the pebble game $\BP_{k+1}(G,H)$.
\end{theorem}

\paragraph{Logic}
There is also a close connection between the Weisfeiler-Leman algorithm and the $k$-variable fragment of first-order logic with counting quantifiers.

As usual first-order logic (\FO) is build inductively starting from the atomic formulas.
The atomic formulas are of the form $x = y$ and $Exy$
(for this description we restrict the vocabulary to $\{E\}$ where $E$ is a 2-ary relation that corresponds to the edge relation of a graph).
First-order formulas are build from the atomic formulas in an inductive way using Boolean operations $\wedge,\vee,\neg$,
existential quantifiers $\exists x \varphi(x)$ and universal quantifiers $\forall x \varphi(x)$.

We define \LC\ to be the extension of \FO\ by counting quantifiers of the form $\exists^{\geq i} x \varphi(x)$.
Such a formula is satisfied if there are at least $i$ distinct vertices $v$ that satisfy the formula $\varphi(x)$.
Moreover, for $k \in \mathbb{N}$, we let $\LL^{k}$ be the $k$-variable fragment of \FO, that is those formulas having at most $k$ distinct variables,
and similarly we let $\LC^{k}$ be the $k$-variable fragment of \LC.

Note that while \FO\ and \LC\ have the same expressive power this is not true for $\LL^{k}$ and $\LC^{k}$.

A sentence is a formula without free variables.
We say two graphs $G$ and $H$ are equivalent with respect to $\LC^{k}$, denoted by $G \equiv_{\LC^k} H$, if for every
sentence $\varphi$ in the logic $\LC^{k}$ it holds that $G \models \varphi$ if and only if $H \models \varphi$.

With this definition we get the following connection between first-order logic with counting quantifiers and the Weisfeiler-Leman algorithm.

\begin{theorem}[\cite{caifurimm92,hel96,immlan90}]
 \label{thm:eq-wl-ck}
 Let $G,H$ be two graphs.
 Then $G \simeq_k H$ if and only if $G \equiv_{\LC^{k+1}} H$.
\end{theorem}

\begin{corollary}
 Let $G$ be a graph that is identified by the $k$-dimensional Weisfeiler-Leman algorithm.
 Then there is a sentence $\phi_G$ of the logic $\LC^{k+1}$ that characterises $G$ up to isomorphism.
\end{corollary}

\begin{proof}
 Let $n \coloneqq |V(G)|$.
 For every $n$-vertex graph $H$ such that $G \not\cong H$ there is a sentence $\psi_H \in \LC^{k+1}$ over variables $x_1,\dots,x_{k+1}$ such that $G \models \psi_H$ and $H \not\models \psi_H$.
 We define
 \[\varphi_G \coloneqq \exists^{\geq n}x_1 (x_1 = x_1) \;\wedge \neg\exists^{\geq n+1}x_1 (x_1 = x_1) \;\wedge \bigwedge_{H\colon |V(H)| = n, G \not\cong H}\psi_H.\qedhere\]
\end{proof}

\subsection{Canonisation}

A common approach to tackle the isomorphism problem is to canonise the input graphs, that is, to compute a standard representation of the input graph
that only depends on the isomorphism type of the graph and not on its representation.
Formally, a graph canonisation can be defined as follows.

\begin{definition}
 A \emph{graph canonisation} for a graph class $\mathcal{C}$ is a function $\kappa\colon \mathcal{C} \rightarrow \mathcal{C}$ such that
 \begin{enumerate}
  \item $\kappa(G) \cong G$ for all $G \in \mathcal{C}$, and
  \item $\kappa(G) = \kappa(H)$ for all graphs $G,H \in \mathcal{C}$ such that $G \cong H$.
 \end{enumerate}
\end{definition}

Note that the isomorphism problem for a class $\mathcal{C}$ easily reduces to computing a graph canonisation for $\mathcal{C}$.
A reduction in the other direction is not known, that is no polynomial-time algorithm is known that reduces the graph canonisation problem for a class $\mathcal{C}$ to the corresponding isomorphism problem.
However, most algorithms for the isomorphism problem that are based on combinatorial approaches can be easily turned into graph canonisation algorithms.
For example, this is true for isomorphism tests that are based on the Weisfeiler-Leman algorithm.

\begin{theorem}
 \label{thm:canon-from-wl}
 Let $\mathcal{C}$ be a graph class and suppose the $k$-dimensional Weisfeiler-Leman algorithm identifies all coloured graphs in $\mathcal{C}$.
 Then there is a graph canonisation for $\mathcal{C}$ that can be computed in time $O(n^{k+3}\log n)$.
\end{theorem}

Here, it is assumed that $\mathcal{C}$ is collection of uncoloured graphs (that is closed under isomorphism), and a coloured graph is contained in $C$ if its uncoloured version is in $\mathcal{C}$.
This theorem is essentially known among people working on the Weisfeiler-Leman algorithm.
Nonetheless we give a proof in Appendix \ref{sec:app-canon-from-wl}.

\section{Split Pairs and Flip Functions}
\label{sec:split}

We first show that the $\ell$-dimensional Weisfeiler-Leman algorithm identifies all graphs of rank width at most $k$ for some $\ell \in O(k)$.
Let $G$ be a graph of rank width $k$.
On a high level, our approach is similar to the proof of the same result for graphs of bounded tree width \cite{gromar99}.
For a set $X \subseteq V(G)$ such that $\rho_G(X) \leq k$ we wish to find a small set of vertices such that pebbling these vertices splits the graph into multiple sets $C$ that can be treated independently.
Moreover, each of these sets $C$ should satisfy that $C \subseteq X$
or $C \subseteq \overline{X}$.
As there may be many edges between $X$ and $\overline{X}$, it is not
obvious how to achieve this. In particular, we cannot simply remove a
few vertices in order to separate $X$ from $\overline X$.
Split pairs and flip functions are our way of dealing with this. 

Let $G$ be a graph and $X \subseteq V(G)$.
For $v,w \in X$ we define $v \approx_X w$ if $N(v) \cap \overline{X} = N(w) \cap \overline{X}$.
For $v \in X$ we define the vector $\xvec_X(v) \coloneqq (a_{v,w})_{w \in \overline{X}} \in \mathbb{F}_2^{\mathbb{\overline{X}}}$
where $a_{v,w} = 1$ if and only if $vw \in E(G)$.
Note that $v \approx_X w$ if and only if $\xvec_X(v) = \xvec_X(w)$.
Moreover, for $S \subseteq X$ we define $\xvec_X(S) \coloneqq \{\xvec_X(v) \mid v \in S\}$.

\begin{lemma}
 \label{la:subset-linear-independence}
 Let $Y \subseteq X \subseteq V(G)$ and suppose $S \subseteq X$ such that $\xvec_X(S)$ is linearly independent.
 Then $\xvec_Y(S \cap Y)$ is linearly independent.
\end{lemma}

\begin{proof}
 We have $\xvec_X(S \cap Y) \subseteq \xvec_X(S)$ and thus, $\xvec_X(S \cap Y)$ is linearly independent.
 Moreover, $\overline{X} \subseteq \overline{Y}$ which means that every vector $\xvec_Y(v) \in \xvec_Y(S \cap Y)$ is an extension of $\xvec_X(v) \in \xvec_X(S \cap Y)$.
 So $\xvec_Y(S \cap Y)$ is also linearly independent.
\end{proof}

For any set of vectors $S \subseteq \mathbb{F}_2^{n}$ we denote by $\langle S \rangle$ the linear space spanned by $S$.
A set $B \subseteq \mathbb{F}_2^{n}$ is a \emph{linear basis for $\langle S \rangle$} if $B$ is linearly independent and $\langle B \rangle = \langle S \rangle$.

\begin{definition}
 Let $G$ be a graph and $X \subseteq V(G)$.
 A pair $(A,B)$ is a \emph{split pair for $X$} if
 \begin{enumerate}
  \item $A \subseteq X$ and $B \subseteq \overline{X}$,
  \item $\xvec_X(A)$ forms a linear basis for
    $\langle\xvec_X(X)\rangle$, and
\item $\xvec_{\overline{X}}(B)$ forms a linear basis for $\langle\xvec_{\overline{X}}(\overline{X})\rangle$.
 \end{enumerate}
\end{definition}

Note that $|A| = \rho_G(X) = \rho_G(\overline{X}) = |B|$.
Also observe that if $(A,B)$ is a split pair for $X$ then $(B,A)$ is a split pair for $\overline{X}$.
As a special case the pair $(\emptyset,\emptyset)$ is defined to be a split pair for $X = V(G)$.
An \emph{ordered split pair for $X$}
is a pair $(\bar a,\bar b) = ((a_1,\dots,a_q),(b_1,\dots,b_p))$ such that $(\{a_1,\dots,a_q\},\{b_1,\dots,b_p\})$ is a split pair for $X$.

\begin{lemma}
 \label{la:same-neighbors-are-equivalent}
 Let $G$ be a graph, $X \subseteq V(G)$ and suppose $(A,B)$ is a split pair for $X$.
 Also let $v,w \in X$ such that  $N(v) \cap B = N(w) \cap B$.
 Then $v \approx_X w$.
 Similarly, $v' \approx_{\overline{X}} w'$ for all $v',w' \in \overline{X}$ such that  $N(v') \cap A = N(w') \cap A$.
\end{lemma}

\begin{proof}
 Let $v,w \in X$ and suppose $B = \{b_1,\dots,b_p\}$.
 Then, for all $i \in [p]$ we have $vb_i\in E(G)$ if and only if $wb_i\in
 E(G)$. Thus 
 \[\left(\xvec_{\overline{X}}(b_i)\right)_v =
   \left(\xvec_{\overline{X}}(b_i)\right)_w,\]
 that is, the $v$-entry of the vector $\xvec_{\overline{X}}(b_i)$
 coincides with the $w$-entry.
 Since $\xvec_{\overline{X}}(B)$ forms a linear basis for $\langle\xvec_{\overline{X}}(\overline{X})\rangle$, we conclude that
 \[\left(\xvec_{\overline{X}}(v')\right)_v = \left(\xvec_{\overline{X}}(v')\right)_w\]
 for all $v' \in \overline{X}$.
 But this means $N(v) \cap \overline{X} = N(w) \cap \overline{X}$ and thus, $v \approx_X w$.
 The second statement is proved analogously.
\end{proof}

For a coloured graph $G = (V,E,\chi)$
and a sequence of vertices $\bar v = (v_1,\dots,v_\ell) \in V^{\ell}$ we define $\chi^{\bar v}$ to be the colouring obtained from $\chi$ after individualising the vertices in $\bar v$ by assigning them the position of their last appearance in $\bar v$ and shifting all other colours accordingly.
More formally,
\[\chi^{\bar v}\colon V \rightarrow \mathbb{N}\colon v \mapsto \begin{cases}
                                                                i &\text{if } v = v_i \wedge \forall j > i \colon v \neq v_j\\
                                                                \chi(v) + \ell &\text{otherwise}
                                                               \end{cases}.
\]
Moreover, we denote by $\chi^{\bar v,G}_{(\infty)}$ the stable colouring obtained from applying the colour refinement algorithm (i.e.\ the $1$-dimensional Weisfeiler-Leman algorithm) to $(G,\chi^{\bar v})$.
As before, we may omit the graph $G$ if it is clear from context and only write $\chi^{\bar v}_{(\infty)}$.

Also, to simplify notation, for tuples $\bar a = (a_1,\dots,a_k)$ and $\bar b = (b_1,\dots,b_\ell)$ we write $(\bar a,\bar b)$ for the tuple $(a_1,\dots,a_k,b_1,\dots,b_\ell)$ obtained from concatenating $\bar a$ and $\bar b$.

\begin{corollary}
 \label{cor:same-colour-are-equivalent}
 Let $G$ be a graph, $X \subseteq V(G)$ and suppose $(\bar a,\bar b)$ is an ordered split pair for $X$.
 Also let $v,w \in X$ such that  $\chi_{(\infty)}^{(\bar a,\bar b)}(v) = \chi_{(\infty)}^{(\bar a,\bar b)}(w)$.
 Then $v \approx_X w$.
 Similarly, $v' \approx_{\overline{X}} w'$ for all $v',w' \in \overline{X}$ such that  $\chi_{(\infty)}^{(\bar a,\bar b)}(v') = \chi_{(\infty)}^{(\bar a,\bar b)}(w')$.
\end{corollary}

We need to argue how to actually split the graph into independent parts using split pairs.
Similar to the previous corollary, we individualise a split pair and perform the colour refinement algorithm.
We claim that this graph consists of independent parts as desired.
In order to make these parts visible we consider the concept of a flip function.

\begin{definition}
 Let $G = (V,E,\chi)$ be a vertex-coloured graph where $\chi\colon V \rightarrow \mathcal{C}$.
 A \emph{flip function for $G$} is a mapping $f\colon \mathcal{C} \times \mathcal{C} \rightarrow \{0,1\}$ such that $f(c,c') = f(c',c)$ for all $c,c' \in \mathcal{C}$.

 Moreover, for a graph $G = (V,E,\chi)$ and a flip function $f$ we define the \emph{flipped graph} $G^{f} = (V,E^{f},\chi)$ where
 \begin{align*}
  E^{f} \coloneqq \;\;\;\;\;\; &\left\{vw \mid vw \in E \wedge f(\chi(v),\chi(w)) = 0\right\} \\
                  \cup \;\;    &\left\{vw \mid v \neq w \wedge vw \notin E \wedge f(\chi(v),\chi(w)) = 1\right\}.
 \end{align*}
\end{definition}

For a coloured graph $G$ and a flip function $f$
we let $\comp(G,f)\subseteq 2^{V(G)}$ be the set of vertex sets of the connected components of $G^{f}$.
Observe that $\comp(G,f)$ forms a partition of the vertex set of $G$.

\begin{lemma}
 \label{la:find-flip-function}
 Let $G=(V,E,\chi)$
 be a coloured graph and $X \subseteq V(G)$.
 Also let $(\bar a,\bar b)$ be an ordered split pair for $X$.
 
 Then there is a flip function $f$ for the graph $G' = (V,E,\chi^{(\bar a,\bar b),G}_{(\infty)})$
 such that for every $C \in \comp(G',f)$ it holds that $C \subseteq X$ or $C \subseteq \overline{X}$.
\end{lemma}

Before diving into the proof, let us briefly discuss the high-level idea.
Consider two colour classes $\widehat{P} = \big(\chi^{(\bar a,\bar b),G}_{(\infty)}\big)^{-1}(c)$ and $\widehat{Q} = \big(\chi^{(\bar a,\bar b),G}_{(\infty)}\big)^{-1}(c')$ (for two colours $c$ and $c'$).
Corollary \ref{cor:same-colour-are-equivalent} implies that the bipartite graph between $\widehat{P} \cap X$ and $\widehat{Q} \cap \overline{X}$ is either empty or complete.
In the latter case, we can simply set $f(c,c') = 1$ to remove all edges between $\widehat{P} \cap X$ and $\widehat{Q} \cap \overline{X}$ in the flipped graph.
Now, we only have to ensure that the bipartite graph induced by $\widehat{Q} \cap X$ and $\widehat{P} \cap \overline{X}$ is complete as well (since edges between these two sets are also flipped).
However, this follows from the fact that $\chi^{(\bar a,\bar b),G}_{(\infty)}$ is stable with respect to the colour refinement algorithm.

\begin{proof}
 Let $\bar u \coloneqq (\bar a,\bar b)$.
 We define the flip function $f$ such that $f(c,c') = 1$ if there are $v \in X$ and $w \in \overline{X}$ such that $vw \in E(G)$
 and $\{ \chi^{\bar u}_{(\infty)}(v),\chi^{\bar u}_{(\infty)}(w)\} = \{c,c'\}$.
 We argue that there are no $v \in X$ and $w \in \overline{X}$ such that $vw$ is an edge in the flipped graph $(G')^{f}$.
 
 Suppose towards a contradiction this statement does not hold, that is, there are $v \in X$ and $w \in \overline{X}$ such that $vw \in E(G^{f})$.
 Let $c = \chi^{\bar u}_{(\infty)}(v)$ and $c' = \chi^{\bar u}_{(\infty)}(w)$.
 Then $vw \notin E(G)$, because if $vw\in E(G)$ then $f(c,c') = 1$ and thus
 $vw\notin E(G^f)$. Moreover, $f(c,c') = 1$, because $vw \notin E(G)$
 and $vw\in E(G^f)$. 
 This means that there are $v' \in X$ and $w' \in \overline{X}$ such that $v'w' \in E(G)$ and $\{\chi^{\bar u}_{(\infty)}(v'),\chi^{\bar u}_{(\infty)}(w')\} = \{c,c'\}$.
 
 Now we distinguish two cases.
 The first is that $\chi^{\bar u}_{(\infty)}(v') = c$ and hence, $\chi^{\bar u}_{(\infty)}(w') = c'$.
 Then $v \approx_X v'$ and $w \approx_{\overline{X}} w'$ by Corollary \ref{cor:same-colour-are-equivalent}.
 But this implies that
 \[vw \in E(G) \;\;\Leftrightarrow\;\; vw' \in E(G) \;\;\Leftrightarrow\;\; v'w' \in E(G)\]
 which is a contradiction.
 
 Let us turn to the second, more complicated, case
 that $\chi^{\bar u}_{(\infty)}(v') = c'$ and $\chi^{\bar u}_{(\infty)}(w') = c$.
 Let $P = (\chi^{\bar u}_{(\infty)})^{-1}(c) \cap X$, $\overline{P} = (\chi^{\bar u}_{(\infty)})^{-1}(c) \cap \overline{X}$,
 $Q = (\chi^{\bar u}_{(\infty)})^{-1}(c') \cap X$ and $\overline{Q} = (\chi^{\bar u}_{(\infty)})^{-1}(c') \cap \overline{X}$.
 So $v \in P$, $v' \in Q$, $w \in \overline{Q}$ and $w' \in \overline{P}$ (see Figure \ref{fig:find-flip-function}).
 
 \begin{figure}
  \centering
  \begin{tikzpicture}
   \node at (-0.8,2.4) {$X$};
   \node at (3.2,2.4) {$\overline{X}$};
   
   \node at (-0.4,0.9) {$P$};
   \node at (4.4,0.9) {$\overline{P}$};
   \node at (-0.4,-0.6) {$Q$};
   \node at (4.4,-0.6) {$\overline{Q}$};
   
   \draw (0,0) ellipse (0.8cm and 2.4cm);
   \draw (4,0) ellipse (0.8cm and 2.4cm);
   
   \draw[red, fill=red, opacity=0.4] (0,0.0) ellipse (0.4cm and 0.6cm);
   \draw[blue, fill=blue, opacity=0.4] (0,1.5) ellipse (0.4cm and 0.6cm);
   \draw[red, fill=red, opacity=0.4] (4,0.0) ellipse (0.4cm and 0.6cm);
   \draw[blue, fill=blue, opacity=0.4] (4,1.5) ellipse (0.4cm and 0.6cm);
   
   \node at (0,-1.2) {{\LARGE $\vdots$}};
   \node at (4,-1.2) {{\LARGE $\vdots$}};
   
   \node[label={[label distance=-10pt]135:$v$}] (v) at (0.1,1.3) {$\bullet$};
   \node[label={[label distance=-10pt]45:$w$}] (w) at (3.9,-0.2) {$\bullet$};
   \node[label={[label distance=-10pt]135:$v'$}] (vp) at (0.1,-0.2) {$\bullet$};
   \node[label={[label distance=-10pt]45:$w'$}] (wp) at (3.9,1.3) {$\bullet$};
   
  \end{tikzpicture}
  \caption{Visualisation of the sets $P$, $\overline{P}$, $Q$ and $\overline{Q}$ from the proof of Lemma \ref{la:find-flip-function}.}
  \label{fig:find-flip-function}
 \end{figure}
 
 \begin{claim}
  \label{claim:non-edge}
  Let $y \in P$ and $z \in \overline{Q}$. Then $yz \notin E(G)$.
 \end{claim}
 \begin{claimproof}
  We have $v \approx_X y$ and $w \approx_{\overline{X}} z$ by Corollary \ref{cor:same-colour-are-equivalent}.
  Hence,
  \[vw \in E(G) \;\;\Leftrightarrow\;\; vz \in E(G) \;\;\Leftrightarrow\;\; yz \in E(G).\qedhere\]
 \end{claimproof}

 \begin{claim}
  \label{claim:edge}
  Let $y \in Q$ and $z \in \overline{P}$. Then $yz \in E(G)$.
 \end{claim}
 \begin{claimproof}
  We have $v' \approx_X y$ and $w' \approx_{\overline{X}} z$ by Corollary \ref{cor:same-colour-are-equivalent}.
  Hence,
  \[v'w' \in E(G) \;\;\Leftrightarrow\;\; v'z \in E(G) \;\;\Leftrightarrow\;\; yz \in E(G).\qedhere\]
 \end{claimproof}
 
 Now $|N(v) \cap Q| = |N(v) \cap (Q \cup \overline{Q})| = |N(w') \cap (Q \cup \overline{Q})| \geq |Q|$ by Claim \ref{claim:non-edge} and \ref{claim:edge}.
 This means $Q \subseteq N(v)$. In particular, $v\in N(v')$.
 It follows from Claim \ref{claim:edge} that $\overline P\subseteq N(v')$.
 Thus $|N(v') \cap (P \cup \overline{P})| \geq
 |\overline{P}| + 1$.
 Since $\chi^{\bar u}_{(\infty)}(v') = \chi^{\bar u}_{(\infty)}(w) = c'$ we conclude that $|N(w) \cap (P \cup \overline{P})| \geq |\overline{P}| + 1$.
 But $|N(w) \cap (P \cup \overline{P})| = |N(w) \cap \overline{P}| \leq |\overline{P}|$ by Claim \ref{claim:non-edge}.
 This is a contradiction.
\end{proof}

To be able to treat the connected components of the flipped graph independently we need to argue that
applying a flip function to two graphs neither changes the isomorphism problem nor the effect of the Weisfeiler-Leman algorithm.

\begin{lemma}
 \label{la:flip-function-isomorphism}
 Let $G,G'$ be two coloured graphs and let $f$ be a flip function for $G$ and $G'$.
 Also let $\varphi\colon V(G) \rightarrow V(G')$ be a bijection.
 Then $\varphi\colon G \cong G'$ if and only if $\varphi\colon G^{f} \cong (G')^{f}$.
\end{lemma}

\begin{proof}
 Trivial.
\end{proof}

\begin{lemma}
 \label{la:flip-function-pebble-game}
 Let $G = (V,E,\chi)$, $G' = (V',E',\chi')$ be two coloured graphs and let $f$ be a flip function for $G$ and $G'$.
 Also let $(\bar v,\bar w) = ((v_1,\dots,v_k),(w_1,\dots,w_k))$ be a position in the $k$-bijective pebble game $\BP_k(G,G')$.
 Then Spoiler wins from $(\bar v,\bar w)$ in $\BP_k(G,G')$ if and only if Spoiler wins from $(\bar v,\bar w)$ in $\BP_k(G^{f},(G')^{f})$.
\end{lemma}

\begin{proof}
 A position $(\bar v,\bar w)$ in the pebble game $\BP_{k}(G,G')$ is a winning position for Spoiler (i.e., the ordered subgraphs induced by $\bar v$ and $\bar w$ are not isomorphic)
 if and only if it is a winning position for Spoiler in the game $\BP_{k}(G^{f},(G')^{f})$.
\end{proof}

For two colourings $\chi,\chi'\colon V \rightarrow \mathcal{C}$ we write $\chi \equiv \chi'$ if $\chi \preceq \chi'$ and $\chi' \preceq \chi$, that is the partitions induced by the colour classes are the same for both colourings.

\begin{corollary}
 \label{cor:colour-refinement-flip-function}
 Let $G = (V,E,\chi)$ be a coloured graph and let $f$ be a flip function for $G$.
 Then $\chi_{(\infty)}^{G} \equiv \chi^{G,f}_{(\infty)}$ where
 $\chi^{G,f}_{(\infty)}$ is the stable colouring computed by colour
 refinement applied to the graph $G^{f}$.
\end{corollary}

\begin{proof}
 It holds that $\chi_{(\infty)}^{G}(v) = \chi_{(\infty)}^{G}(w)$ if and only if Spoiler wins from position $((v,v),(w,w))$ in the game $\BP_{2}(G,G)$ \cite{caifurimm92,hel96,immlan90}.
 So the statement follows from Lemma \ref{la:flip-function-pebble-game}.
\end{proof}

For $\bar v = (v_1,\dots,v_k) \in V^{k}$ and $C \subseteq V$ we define the tuple $\bar v \cap C = (v_i)_{i \in I}$ where $I = \{i \in [k] \mid v_i \in C\}$.
Also, for a second tuple $\bar w = (w_1,\dots,w_\ell) \in V^{\ell}$, we write $\bar v \subseteq \bar w$ if $\{v_1,\dots,v_k\} \subseteq \{w_1,\dots,w_\ell\}$.

\begin{corollary}
 \label{cor:remove-pebbles-from-outside-component}
 Let $G = (V,E,\chi)$, $G' = (V',E',\chi')$ be two coloured graphs and let $f$ be a flip function for $G$ and $G'$.
 Let $\bar v \in V^{k}$ and $\bar v' \in (V')^{k}$.
 Let $C$ be a connected component of $G^{f}$ such that $\chi(u) \neq \chi(w)$ for all $u \in C$, $w \in V\setminus C$,
 and let $C'$ a connected component of $(G')^{f}$ such that $\chi'(u') \neq \chi'(w')$ for all $u' \in C'$, $w' \in V'\setminus C'$.
 Suppose that \[(G[C],\chi^{\bar v,G}_{(\infty)}) \not\cong
 (G'[C'],\chi^{\bar v',G'}_{(\infty)}).\]
 Let $\bar w = \bar v \cap C$ and $\bar w' = \bar v' \cap C'$.
 Then \[(G[C],\chi^{\bar w,G}_{(\infty)}) \not\cong (G'[C'],\chi^{\bar w',G'}_{(\infty)})\]
 or $(G,\chi^{\bar v}) \not\simeq_1 (G',(\chi')^{\bar v'})$.
\end{corollary}

\begin{proof}
 Suppose $\bar v = (v_1,\dots,v_k)$ and $\bar v' = (v_1',\dots,v_k')$.
 Let $I \coloneqq \{i \in [k] \mid v_i \in C\}$ and $I' \coloneqq \{i \in [k] \mid v_i' \in C'\}$.
 Suppose that $(G,\chi^{\bar v}) \simeq_1 (G',(\chi')^{\bar v'})$.
 Then $(G^{f},\chi^{\bar v}) \simeq_1 ((G')^{f},(\chi')^{\bar v'})$ by Lemma \ref{la:flip-function-pebble-game} and Theorem \ref{thm:eq-wl-pebble} and thus, $I = I'$.
 Now suppose \[\varphi\colon (G[C],\chi^{\bar w,G}_{(\infty)}) \cong (G'[C'],\chi^{\bar w',G'}_{(\infty)}).\]
 Since $I = I'$ it follows that
 \[\varphi\colon (G[C],\chi^{\bar v}) \cong (G'[C'],(\chi')^{\bar v'}).\]
 Now a simple inductive argument gives that
 \[\varphi\colon (G[C],\chi^{\bar v,G}_{i}) \cong (G'[C'],\chi^{\bar v',G'}_{i})\]
 for all $i \in \mathbb{N}$ since colour refinement only takes colours of neighbours into account.
 Note that there is no difference between performing colour refinement on $G$ (resp.\ $G'$) or $G^{f}$ (resp.\ $(G')^{f}$) by Corollary \ref{cor:colour-refinement-flip-function}.
\end{proof}

\section{Weisfeiler-Leman for Graphs of Bounded Rank Width}
\label{sec:WLdim}
In this section we give a proof of Theorem \ref{theo:1}.
The basic strategy for the proof is simple.
Given two non-isomorphic graphs $G$ and $H$, where $G$ has rank width at most $k$, we give a winning strategy for Spoiler in the game $\BP_{\ell}(G,H)$ for $\ell = 3k+5$.
Spoiler's strategy in the game is to play along a rank decomposition $(T,\gamma)$ for the graph $G$.
At a specific node $t \in V(T)$ of the rank decomposition, Spoiler plays an ordered split pair $(\bar a,\bar b)$ for the set $\gamma(t)$ and identifies some component $C$ (with respect to some flip function)
that is different from the corresponding component (specified by the bijection chosen by Duplicator) in the second graph.
In order to distinguish these components, Spoiler continues to play along the rank decomposition going down the tree.
A crucial step to realise this strategy is to ensure that we can remove the pebbles from an ordered split pair of $t$ once Spoiler has pebbled ordered split pairs of the children of $t$.
Towards this end, we introduce the notion of \emph{nice} (triples of) split pairs.

For sets $X,X_1,X_2$ we write $X = X_1 \uplus X_2$ to denote that $X$ is the disjoint union of $X_1$ and $X_2$, that is, $X = X_1 \cup X_2$ and $X_1 \cap X_2 = \emptyset$.

\begin{definition}
 \label{def:nice-split-pairs}
 Let $G$ be a graph and $X,X_1,X_2 \subseteq V(G)$ such that $X = X_1 \uplus X_2$.
 Let $(A,B)$ be a split pair for $X$ and let $(A_i,B_i)$ be a split pair for $X_i$, $i \in \{1,2\}$.
 We say that $(A_i,B_i)$, $i \in \{1,2\}$, are \emph{nice (with respect to $(A,B)$)} if
 \begin{enumerate}
  \item $A \cap X_i \subseteq A_i$, and
  \item $B_{3-i} \cap X_i \subseteq A_i$
 \end{enumerate}
 for both $i \in \{1,2\}$.
\end{definition}

Naturally, a triple of ordered split pairs is nice if the underlying unordered triple of split pairs is nice.

\begin{lemma}
 \label{la:find-nice-split-pairs}
 Let $G$ be a graph and $X,X_1,X_2 \subseteq V(G)$ such that $X = X_1 \uplus X_2$.
 Let $(A,B)$ be a split pair for $X$.
 Then there are nice split pairs $(A_i,B_i)$ for $X_i$, $i \in
 \{1,2\}$, such that additionally $B_i \cap \overline{X} \subseteq B$.
\end{lemma}

\begin{proof}
 We first pick $A_i$ for both $i \in \{1,2\}$.
 Since $X_i \subseteq X$ we can choose $A_i$ in such a way that $A \cap X_i \subseteq A_i$ by Lemma \ref{la:subset-linear-independence}.
 
 The set $\xvec_{X_{3-i}}(A_{3-i})$ spans every element in the set $\xvec_{X_{3-i}}(X_{3-i}) \subseteq \mathbb{F}_2^{X_i \cup \overline{X}}$.
 Hence, $\xvec_{\overline{X_{i}}}(A_{3-i})$ spans every element in the set $\xvec_{\overline{X_{i}}}(X_{3-i}) \subseteq \mathbb{F}_2^{X_{i}}$.
 
 Moreover, the set $\xvec_{\overline{X}}(B)$ spans every element in the set $\xvec_{\overline{X}}(\overline{X}) \subseteq \mathbb{F}_2^{X_1 \cup X_2}$.
 So $\xvec_{\overline{X_i}}(B)$ spans every element in the set $\xvec_{\overline{X_i}}(\overline{X}) \subseteq \mathbb{F}_2^{X_i}$.
 
 Together this means that $\xvec_{\overline{X_i}}(B \cup A_{3-i})$ spans every element in the set $\xvec_{\overline{X_i}}(\overline{X_i}) \subseteq \mathbb{F}_2^{X_i}$.
 We choose $B_i \subseteq B \cup A_{3-i}$ inclusionwise maximal
 such that $\xvec_{\overline{X_i}}(B_i)$ is linearly independent.
\end{proof}

We remark that the additional guarantee $B_i \cap \overline{X} \subseteq B$ is not relevant to obtain a linear upper bound for the Weisfeiler-Leman dimension of graphs of rank width at most $k$.
However, the additional overlap between the sets allows us to improve on the constant factors appearing in our arguments.

Also, we shall need the following simple observation.
Let $G$ be a graph.
A \emph{component partition of $G$} is a partition $\mathcal{P}$ of $V(G)$ such that every connected component of $G$ is contained in one block of $\CP$, i.e., for every connected component $C$ of $G$ there is some $P \in \mathcal{P}$ such that $C \subseteq P$. 

\begin{observation}
 \label{obs:component-bijection}
 Let $G,H$ be two non-isomorphic graphs and let $\mathcal{P},\mathcal{Q}$ be component partitions of $G$ and $H$, respectively.
 Also let $\sigma\colon V(G) \rightarrow V(H)$ be any bijection.
 Then there is some $v \in V(G)$ such that $G[P] \not\cong H[Q]$ where $P \in \mathcal{P}$ is the unique set such that $v \in P$ and $Q \in \mathcal{Q}$ is the unique set such that $\sigma(v) \in Q$.
\end{observation}

\begin{proof}
 Since $G \not\cong H$ there is some graph $F$ such that
 \[|\{P \in \mathcal{P} \mid G[P] \cong F\}| > |\{Q \in \mathcal{Q} \mid H[Q] \cong F\}|.\]
 Let $\mathcal{P}_F \coloneqq \{P \in \mathcal{P} \mid G[P] \cong F\}$ and $\mathcal{Q}_F \coloneqq \{Q \in \mathcal{Q} \mid H[Q] \cong F\}$.
 Also let $\widehat{\mathcal{P}}_F \coloneqq \bigcup_{P \in \mathcal{P}_F}P$ and $\widehat{\mathcal{Q}}_F \coloneqq \bigcup_{Q \in \mathcal{Q}_F}Q$.
 Then
 \[|\widehat{\mathcal{P}}_F| > |\widehat{\mathcal{Q}}_F|.\]
 So there is some $v \in \widehat{\mathcal{P}}_F$ such that $\sigma(v) \notin \widehat{\mathcal{Q}}_F$.
 Hence, $G[P] \not\cong H[Q]$ where $P \in \mathcal{P}$ is the unique set such that $v \in P$ and $Q \in \mathcal{Q}$ is the unique set such that $\sigma(v) \in Q$.
\end{proof}

\begin{theorem}[Theorem \ref{theo:1} restated]
 \label{thm:wl-dimension-rw}
 The $(3k+4)$-dimensional Weisfeiler-Leman algorithm identifies every graph of rank width at most $k$.
\end{theorem}

\begin{proof}
 Let $G = (V(G),E(G),\chi_G)$ be a graph such that $\rw(G) \leq k$ and moreover let $H = (V(H),E(H),\chi_H)$ be a second graph such that $G \not\cong H$.
 Let $(T,\gamma)$ be a rank decomposition of width $k$ for the graph $G$.
 
 We argue that Spoiler wins the bijective $\ell$-pebble game played over graphs $G$ and $H$ where $\ell = 3k + 5$.
 In combination with Theorem \ref{thm:eq-wl-pebble} this proves the theorem.
 Actually, we first give a winning strategy for Spoiler that requires $\ell = 6k + 5$ many pebbles.
 Then we proceed to argue how to realise this strategy using only $3k + 5$ many pebbles.
 
 On a high-level, Spoiler's strategy is to play along the rank decomposition $(T,\gamma)$ and ``confine the non-isomorphism'' to smaller and smaller parts of $G$ and $H$.
 More precisely, for a node $t \in V(T)$, Spoiler's idea is to pebble an ordered split pair $(\bar a,\bar b)$ of $X = \gamma(t)$.
 Let $f$ be the flip function obtained from Lemma \ref{la:find-flip-function} with respect to $X$.
 To ``confine the non-isomorphism'' to a subset of $X$, Spoiler identifies non-isomorphic components $C \subseteq X$ and $C' \subseteq V(H)$ in the flipped graphs $G^f$ and $H^f$ (after individualising a split pair and performing colour refinement).
 To remember the components $C$ and $C'$, Spoiler places additional ``component marker'' pebbles on some vertices $v \in C$ and $v' \in C'$.
 To eventually reach a winning position, Spoiler's idea is to find such components for nodes $t$ which are further and further away from the root of $T$, eventually ending up at a leaf of $T$ at which point $|C| = 1$ and Spoiler has a simple winning strategy.
 In the following, we describe how this high-level strategy can be realized.
 
 For a node $t \in V(T)$ a tuple $(\bar a,\bar b)$ is an \emph{ordered split pair for $t$} if $(\bar a,\bar b)$ is an ordered split pair for $\gamma(t)$.
 
 Now suppose the play is at a position $((\bar a,\bar b,v),(\bar a',\bar b',v'))$
 such that the following conditions are satisfied:
 \begin{itemize}
  \item There is a node $t \in V(T)$ such that $(\bar a,\bar b)$ is an ordered split pair for $t$.
  \item $v \in \gamma(t)$.
  \item Let $f$ be the flip function obtained from Lemma \ref{la:find-flip-function} with respect to $X = \gamma(t)$.
   Let $C \in \comp((G,\chi^{(\bar a,\bar b)}_{(\infty)}),f)$ such that $v \in C$.
   Similarly let $C' \in \comp((H,\chi^{(\bar a',\bar b')}_{(\infty)}),f)$ such that $v' \in C'$.
   Then \[\Big(G[C],\chi^{(\bar a,\bar b,v)}_{(\infty)}\Big) \not\cong \Big(H[C'],\chi^{(\bar a',\bar b',v')}_{(\infty)}\Big).\]
 \end{itemize}
 Note that $C \subseteq X$ by Lemma \ref{la:find-flip-function}.
 Also observe that $\chi^{(\bar a,\bar b,v)}_{(\infty)}(u_1) \neq \chi^{(\bar a,\bar b,v)}_{(\infty)}(u_2)$ for all $u_1 \in C$ and $u_2 \in V(G) \setminus C$.
 This is clear for the graph $G^{f}$ since $v \in C$ and $C$ forms a connected component in $G^{f}$ and thus, it also holds for $G$ by Corollary \ref{cor:colour-refinement-flip-function}.
 The same statement holds for the set $C'$ in the graph $H$.
 
 Initially it is easy for Spoiler to reach such a position for the root node of $T$ (choosing the empty split pair $((),())$ and using Observation \ref{obs:component-bijection}).
 Also observe that in a position as described above the number of pebbles is at most $2k + 1$.
 We now prove by induction on $|\gamma(t)|$ that Spoiler wins from such a position.
 
 For the base step suppose that $|\gamma(t)| = 1$. In this case $C =
 \{v\}$ and Spoiler easily wins using two additional pebbles.
 Recall that the sets $C$ and $C'$ can be recognised by colour refinement since one of the vertices in each set is individualised (cf.\ Corollary \ref{cor:colour-refinement-flip-function}).
 
 So for the inductive step suppose $|\gamma(t)| > 1$.
 Let $t_1$ and $t_2$ be the children of $t$.
 Let $X \coloneqq \gamma(t)$ and $X_i \coloneqq \gamma(t_i)$ for $i \in \{1,2\}$.
 Note that $X = X_1 \uplus X_2$.

 Let $(\bar a_i,\bar b_i)$, $i \in \{1,2\}$, be nice ordered split pairs for $t_i$ (cf.\ Lemma \ref{la:find-nice-split-pairs}).
 Now Spoiler plays pebbles on $(\bar a_1,\bar b_1,\bar a_2,\bar b_2)$ and let $(\bar a_1',\bar b_1',\bar a_2',\bar b_2')$ be Duplicator's answer.
 We also define $\bar\alpha \coloneqq (\bar a,\bar b,\bar a_1,\bar b_1,\bar a_2,\bar b_2,v)$ and $\bar\alpha' \coloneqq (\bar a',\bar b',\bar a_1',\bar b_1',\bar a_2',\bar b_2',v')$.
 On an intuitive level, the advantage of pebbling nice ordered split pairs is that, for $i \in \{1,2\}$,
 we can remove the pebbles $(\bar a,\bar b)$ and $(\bar a_{3-i},\bar b_{3-i})$ without unpebbling some element from $X_i$.
 
 Let $f_i$ be the flip function obtained from Lemma \ref{la:find-flip-function} with respect to the ordered split pair $(\bar a_i,\bar b_i)$ and the set $X_i$.
 Now Spoiler wishes to play another pebble.
 Let $\sigma\colon V(G) \rightarrow V(H)$ be the bijection chosen by Duplicator.
 Without loss of generality we can assume that
 \begin{enumerate}[label = (\alph*)]
  \item $\sigma(\bar \alpha) = \bar \alpha'$, and
  \item $\sigma(C) = C'$
 \end{enumerate}
 (otherwise, Spoiler wins the game using two additional pebbles).
 Additionally, we can assume without loss of generality that $v \in X_1$ (otherwise we swap the roles of $X_1$ and $X_2$).
 
 \begin{figure}
  \centering
  \begin{tikzpicture}
   \node at (0.3,5.5) {$\overline{X}$};
   \node at (6.7,5.5) {$X$};
   \node at (3.3,4.85) {$X_2$};
   \node at (3.3,0.1) {$X_1$};
   
   \node at (6.2,4.15) {$C$};
   \node at (3.4,0.9) {$C_1$};
   \node at (3.4,1.9) {$C_2$};
   \node at (1.0,3.1) {$C_3$};
   \node at (1.0,4.1) {$C_4$};
   
   \draw[line width=2pt,gray!50] (3,-0.2) -- (3,5.2);
   \draw[line width=2pt, dashed, gray!50] (3,2.5) -- (7,2.5);
   \draw[thick] (0,-0.2) rectangle (7,5.2);
   
   \draw[rounded corners] (3.1,0.4) rectangle (6.9,1.35);
   \draw[rounded corners] (3.1,1.45) rectangle (6.9,2.4);
   \draw[rounded corners] (0.7,2.6) rectangle (6.9,3.55);
   \draw[rounded corners] (0.7,3.65) rectangle (6.9,4.6);
   
   \draw[blue, rounded corners = 0.5cm, fill=blue, opacity=0.4] (3.7,0.5) -- (3.7,4.5) -- (4.7,4.5) -- (4.7,3.0) -- (5.7,3.0) -- (5.7,4.5) -- (6.7,4.5) -- (6.7,0.5) -- (5.7,0.5) -- (5.7,2.0) -- (4.7,2.0) -- (4.7,0.5) -- cycle;
   
   \node[label={[label distance=-5pt]0:$v$}] (v) at (4.2,2.1) {$\bullet$};
   \node[label={[label distance=-5pt]0:$w$}] (w) at (6.0,1.1) {$\bullet$};
  \end{tikzpicture}
  \caption{Visualisation for the induction step (Case 1) in the proof of Theorem \ref{thm:wl-dimension-rw}.}
  \label{fig:descent-step}
 \end{figure}
 
 First consider the flip function $f_1$ (see Figure \ref{fig:descent-step}).
 Let $\{C_1,\dots,C_p\} = \{D \in \comp(G,f_1) \mid D \cap C \neq \emptyset\}$.
 Similarly, let $\{C_1',\dots,C_{p'}'\} = \{D' \in \comp(H,f_1) \mid D' \cap C' \neq \emptyset\}$.
 Clearly,
 \[\Big(\big(G^{f_1}\big)[C],\chi^{\bar\alpha}_{(\infty)}\Big) \not\cong \Big(\big(H^{f_1}\big)[C'],\chi^{\bar\alpha'}_{(\infty)}\Big)\]
 using Lemma \ref{la:flip-function-isomorphism}.
 By Observation \ref{obs:component-bijection} there is some $w \in C$ such that
 \[\Big(\big(G^{f_1}\big)[C \cap C_i],\chi^{\bar\alpha}_{(\infty)}\Big) \not\cong \Big(\big(H^{f_1}\big)[C' \cap C_{i'}'],\chi^{\bar\alpha'}_{(\infty)}\Big)\]
 where $i \in [p]$ is the unique index such that $w \in C_i$ and $i' \in [p']$ is the unique index such that $\sigma(w) \in C_{i'}'$.
 Without loss of generality suppose that $i = i' = 1$.
 Applying Lemma \ref{la:flip-function-isomorphism} once again, we get that
 \[\Big(G[C \cap C_1],\chi^{\bar\alpha}_{(\infty)}\Big) \not\cong \Big(H[C' \cap C_{1}'],\chi^{\bar\alpha'}_{(\infty)}\Big).\]
 Also, by Lemma \ref{la:find-flip-function}, it holds that $C_1 \subseteq X_1$ or $C_1 \subseteq \overline{X_1}$.
 In particular, $C \cap C_1 \subseteq X_1$ or $C \cap C_1 \subseteq X_2$.
 \begin{cs}
  \case{1} $C \cap C_1 \subseteq X_1$.\\
  Observe that $\chi^{\bar \alpha}_{(\infty)}(u_1) \neq \chi^{\bar \alpha}_{(\infty)}(u_2)$ for all $u_1 \in C$ and $u_2 \in V(G) \setminus C$.
  Hence, we get that
  \[\Big(G[C_1],\chi^{\bar \alpha}_{(\infty)}\Big) \not\cong
    \Big(H[C_1'],\chi^{\bar \alpha'}_{(\infty)}\Big).\]
  Now Spoiler plays the next pebble as follows: if $v \in C_1$ and $v' \in C_1'$ then he plays $z = v$ and $z' = v'$, otherwise Spoiler plays $z = w$ and $z' = \sigma(w)$.
  Clearly,
  \[\Big(G[C_1],\chi^{(\bar \alpha,z)}_{(\infty)}\Big) \not\cong
    \Big(H[C_1'],\chi^{(\bar \alpha',z')}_{(\infty)}\Big)\]
  Now consider again the flip function $f_1$.
  In $G^{f_1}$ the set $C_1$ forms a connected component and similarly, in $H^{f_1}$ the set $C_1'$ forms a connected component.
  Hence, removing any pebbles from vertices outside $C_1$ (resp.\ $C_1'$) does not affect the stable colouring restricted to the component $C_1$ (resp.\ $C_1'$) by Corollary \ref{cor:colour-refinement-flip-function}.
  Since all pebbles $(\bar a,\bar b,\bar a_{2},\bar b_{2},v)$ (resp.\ $(\bar a',\bar b',\bar a_{2}',\bar b_{2}',v')$) are either outside of $C_1$
  or the corresponding vertices are also pebbled using $(\bar a_1,\bar b_1,z)$ (resp.\ $(\bar a_1',\bar b_1',z')$), we can remove the pebbles 
  $(\bar a,\bar b,\bar a_{2},\bar b_{2},v)$ and $(\bar a',\bar b',\bar a_{2}',\bar b_{2}',v')$ and still get that
  \[\Big(G[C_1],\chi^{(\bar a_1,\bar b_1,z)}_{(\infty)}\Big) \not\cong \Big(H[C_1'],\chi^{(\bar a_1',\bar b_1',z')}_{(\infty)}\Big)\]
  by Corollary \ref{cor:remove-pebbles-from-outside-component} (or Spoiler wins using two additional pebbles).
  But now we can apply the induction hypothesis to $t_1$.
  As a result, Spoiler wins from the current position and hence, Spoiler wins from position $((\bar a,\bar b,v),(\bar a',\bar b',v'))$.
  
  \case{2} $C \cap C_1 \subseteq X_2$.\\
  Let us first remark that this case is not symmetric to the first case since the set $C_1$ is defined with respect to the flip function $f_1$.
  Also, for ease of notation, define $M \coloneqq C \cap C_1$ and $M' \coloneqq C' \cap C_1'$.
  
  First Spoiler plays the next pebble on $w$ and $w' \coloneqq \sigma(w)$.
  Observe that $\chi^{(\bar \alpha,w)}_{(\infty)}(u_1) \neq \chi^{(\bar \alpha,w)}_{(\infty)}(u_2)$ for all $u_1 \in M$ and $u_2 \in V(G) \setminus M$.
  
  Now consider the flip function $f_2$.
  Spoiler wishes to play the next pebble.
  Again, let $\sigma\colon V(G) \rightarrow V(H)$ denote the bijection chosen by Duplicator.
  Without loss of generality we can assume that
  \begin{enumerate}[label = (\alph*)]
   \item $\sigma(\bar \alpha) = \bar \alpha'$,
   \item $\sigma(w) = w'$,
   \item $\sigma(M) = M'$.
  \end{enumerate}
  (as before, otherwise Spoiler wins the game using two additional pebbles).
  Consider the flip function $f_2$.
  Let $\{D_1,\dots,D_q\} = \{D \in \comp(G,f_2) \mid D \cap M \neq \emptyset\}$.
  Similarly, let $\{D_1',\dots,D_{q'}'\} = \{D' \in \comp(H,f_2) \mid D' \cap M' \neq \emptyset\}$.
  Clearly,
  \[\Big(\big(G^{f_2}\big)[M],\chi^{(\bar \alpha,w)}_{(\infty)}\Big) \not\cong \Big(\big(H^{f_2}\big)[M'],\chi^{(\bar \alpha',w')}_{(\infty)}\Big)\]
  using Lemma \ref{la:flip-function-isomorphism}.
  By Observation \ref{obs:component-bijection} there is some $z \in M$ such that
  \[\Big(\big(G^{f_2}\big)[M \cap D_i],\chi^{(\bar \alpha,w)}_{(\infty)}\Big) \not\cong \Big(\big(H^{f_2}\big)[M' \cap D_{i'}'],\chi^{(\bar \alpha',w')}_{(\infty)}\Big)\]
  where $i \in [q]$ is the unique index such that $z \in D_i$ and $i' \in [q']$ is the unique index such that $\sigma(z) \in D_{i'}'$.
  As before, we may assume without loss of generality that $i = i' = 1$.
  Applying Lemma \ref{la:flip-function-isomorphism} once again, we get that
  \[\Big(G[M \cap D_1],\chi^{(\bar \alpha,w)}_{(\infty)}\Big) \not\cong \Big(H[M' \cap D_1'],\chi^{(\bar \alpha',w')}_{(\infty)}\Big).\]
  Now recall that $\chi^{(\bar \alpha,w)}_{(\infty)}(u_1) \neq \chi^{(\bar \alpha,w)}_{(\infty)}(u_2)$ for all $u_1 \in M$ and $u_2 \in V(G) \setminus M$.
  This means that 
  \[\Big(G[D_1],\chi^{(\bar \alpha,w)}_{(\infty)}\Big) \not\cong \Big(H[D_1'],\chi^{(\bar \alpha',w')}_{(\infty)}\Big).\]
  Now Spoiler plays the next pebble as follows: if $w \in D_1$ and $w' \in D_1'$ then he plays $x = w$ and $x' = w'$, otherwise Spoiler plays $x = z$ and $x' = \sigma(z)$.
  Clearly,
  \[\Big(G[D_1],\chi^{(\bar \alpha,w,x)}_{(\infty)}\Big) \not\cong
    \Big(H[D_1'],\chi^{(\bar \alpha',w',x')}_{(\infty)}\Big)\]
  Now consider again the flip function $f_2$.
  In $G^{f_2}$ the set $D_1$ forms a connected component and similarly, in $H^{f_2}$ the set $D_1'$ forms a connected component.
  Hence, removing any pebbles from vertices outside $D_1$ (resp.\ $D_1'$) does not affect the stable colouring restricted to the component $D_1$ (resp.\ $D_1'$) by Corollary \ref{cor:colour-refinement-flip-function}.
  Since all pebbles $(\bar a,\bar b,\bar a_{1},\bar b_{1},v,w)$ (resp.\ $(\bar a',\bar b',\bar a_{1}',\bar b_{1}',v',w')$) are either outside of $D_1$ (recall that $v \in X_1$ and hence, $v \notin D_1$)
  or the corresponding vertices are also pebbled using $(\bar a_2,\bar b_2,x)$ (resp.\ $(\bar a_2',\bar b_2',x')$), we can remove the pebbles 
  $(\bar a,\bar b,\bar a_{1},\bar b_{1},v,w)$ and $(\bar a',\bar b',\bar a_{1}',\bar b_{1}',v',w')$ and still get that
  \[\Big(G[D_1],\chi^{(\bar a_2,\bar b_2,x)}_{(\infty)}\Big) \not\cong \Big(H[D_1'],\chi^{(\bar a_2',\bar b_2',x')}_{(\infty)}\Big)\]
  by Corollary \ref{cor:remove-pebbles-from-outside-component} (or Spoiler wins using two additional pebbles).
  But now we can apply the induction hypothesis to $t_2$.
  As a result, Spoiler wins from the current position and hence,
  Spoiler wins from position $((\bar a,\bar b,v),(\bar a',\bar b',v'))$.
 \end{cs}
 
 Overall, by the induction principle, this gives us a winning strategy for Spoiler in the pebble game played over the graphs $G$ and $H$.
 It remains to analyse the number of pebbles required to implement this strategy.
 Looking at Spoiler's strategy, it is not difficult to see that it requires at most $6k + 5$ many pebbles.
 More precisely, Spoiler needs $6k$ pebbles to pebble the three ordered split pairs $(\bar a, \bar b)$ and $(\bar a_i, \bar b_i)$ for $i \in \{1,2\}$.
 The base step requires three additional pebbles. In the inductive step, five additional pebbles suffice, three for pebbling $v$, $w$ and $x$
 and two pebbles to simulate colour refinement in case the bijections chosen by Duplicator do not match up.
 
 However, taking a closer look, some vertices are always pebbled multiple times due to the nice ordered split pairs.
 In particular, we get that $\bar a \subseteq \bar a_1 \cup \bar a_2$.
 Since there is no need to pebble any vertex multiple times, we conclude that Spoiler can also win using only $5k + 5$ many pebbles.
 
 But even this number can be further improved.
 Indeed, Spoiler can also find nice ordered split pairs $(\bar a, \bar b)$ and $(\bar a_i, \bar b_i)$ for $i \in \{1,2\}$
 such that additionally $\bar b_i \cap \overline{X} \subseteq \bar b$ (cf.\ Lemma \ref{la:find-nice-split-pairs}).
 Then $\bar b_i \subseteq \bar b \cup \bar a_{3-i}$ for both $i \in \{1,2\}$.
 Again, there is no need to pebble any vertex multiple times and hence, Spoiler actually requires only $3k + 5$ many pebbles.
\end{proof}

\section{Capturing PTIME on Graphs of Bounded Rank Width}
\label{sec:cap}
In this section, we prove Theorem~\ref{theo:2}. We start with a quick
introduction to the necessary background from descriptive complexity theory.

\subsection{Preliminaries from Descriptive Complexity Theory} 
\label{sec:descriptive}
We assume that the reader has a solid background in logic and, in
particular, is familiar with the standard fixed-point logics used in
finite model theory. For background and precise definitions, we
refer the reader to the textbooks~\cite{ebbflu99,grakollib+07,gro17,imm99,lib04}.

\subsubsection*{Relational Structures}
We work with finite structures over a relational vocabulary $\tau$. The
universe of a structure $A$ is denoted by $V(A)$, and the
interpretation of a $k$-ary relation symbol $R$ is denoted by
$R(A)$. In particular, we view graphs as structures of vocabulary
$\{E\}$ for a binary relation symbol $E$. For a structure $A$ and a
subset $U\subseteq V(A)$, the \emph{induced substructure} $A[U]$ is the
structure with universe $V(A[U])\coloneqq U$ and relations $R(A[U])\coloneqq R(A)\cap U^k$ for
every $k$-ary relation symbol in the vocabulary. Two structures $A,B$
of the same vocabulary $\tau$ are \emph{isomorphic} ($A\cong
B$) if there is a bijective mapping, called an \emph{isomorphism}, $f:V(A)\to V(B)$ such that for all
$k$-ary $R\in\tau$ and all $\bar a\in V(A)^k$ we have $\bar a\in
R(A)$ if and only if $f(\bar a)\in R(B)$. We extend isomorphisms to structures with
\emph{individualised elements}: for tuples $\bar
v=(v_1,\ldots,v_\ell)\in V(A)^\ell$, $\bar
w=(w_1,\ldots,w_\ell)\in V(B)^\ell$, an \emph{isomorphism}  from
$(A,\bar v)$ to $(B,\bar
w)$ is an isomorphism $f$ from $A$ to $B$ such that $f(v_i)=w_i$ for
all $i \in [\ell]$. We write $(A,\bar v)\cong(B,\bar w)$ if such an isomorphism
exists.

\subsubsection*{Fixed-Point Logic with Counting}

\emph{Inflationary fixed-point logic} is the extension of first-order
logic by a fixed-point operator with an inflationary
semantics. Instead of giving a formal definition of its syntax (where
we follow \cite{gro17}) and
semantics, we give one illustrative
example.

\begin{example}\label{exa:conn}
 The \IFP-sentence 
 \[
  \formel{conn}\coloneqq\forall x_1\forall x_2\;\ifp\Big(X\gets (x_1,x_2)\Bigmid
                  x_1=x_2 \vee E(x_1,x_2) \,\vee \exists x_3\big(X(x_1,x_3)\wedge X(x_3,x_2)\big)\Big)(x_1,x_2)
 \]
 states that a graph is connected.  
\end{example}

\IFP-formulas have individual variables, ranging over the elements of
the universe of a structure, and relation variables, each with a
prescribed arity, ranging over relations of this arity over the
universe. We
write $\phi(X_1,\ldots,X_k,x_1,\ldots,x_\ell)$ to denote that the free
relation variables of a formula are among $X_1,\ldots,X_k$ and the
free individual variables are among $x_1,\ldots,x_\ell$. For a
structure $A$, relations $R_1,\ldots,R_k$ of the appropriate arities,
and elements $v_1,\ldots,v_\ell$, we write $A\models\phi(R_1,\ldots,R_k,v_1,\ldots,v_k)$ to denote that $A$
satisfies $\phi$ if the $X_i$ are interpreted by $R_i$ and the $x_j$ are
interpreted by $v_j$.

\emph{Inflationary fixed-point logic with counting}, \FPC, is the
extension of \IFP\ by counting operators that allow it to speak about
cardinalities of definable sets and relations. To define \FPC, we
interpret the logic \IFP\ over two-sorted extensions of structures by
a numerical sort. For a structure $A$, we let $N(A)$ be the initial
segment $\big\{0,\ldots,|V(A)|\big\}$ of the nonnegative integers. We
let $A^+$ be the two-sorted structure $A\cup (N(A),\le)$, where $\le$
is the natural linear order on $N(A)$. To avoid confusion, we always
assume that $V(A)$ and $N(A)$ are disjoint. 

In a structure $A$, individual variables of the logic \FPC\ range either
over the set $V(A)$ (\emph{vertex variables}) or over the set $N(A)$
(\emph{number variables}). Relation variables may range over mixed
relations, having certain places for vertices and certain places for
numbers. The logic $\FPC$ has all the constructors of $\IFP$, and in
addition \emph{counting terms}
of the form $\#x\;\phi$, where $x$ is a vertex variable and $\phi$
some formula. The value of this term in a structure $A$ is the number
of $v\in V(A)$ such that $A$ satisfies $\phi$ if $x$ is interpreted by
$v$ (under some fixed assignment to the other free variables of
$\phi$).

\begin{example}
  We start by giving an a \FPC-formula $\formel{even}(y)$ with one free
  number variable $y$ stating that $y$ is an even number:
  \[
   \formel{even}(y)\coloneqq\ifp\Big(Y\gets y\Bigmid \forall y' y\le y'\vee \exists y''\exists y'\big(Y(y')\wedge
    \formel{succ}(y',y'')\wedge \formel{succ}(y'',y)\big)\Big)(y),
  \]
  where $\formel{succ}(z,z')\coloneqq z\le z'\wedge\neg z=z'\wedge\forall z''(z''\le
  z\vee z'\le z'')$.
  Then the following \FPC-sentence defines the class of Eulerian
  graphs (that is, graphs with a cyclic walk that traverses all edges,
  which are well-known to be exactly the connected graphs in which all
  vertices have even degree):
  \[
     \formel{eulerian}\coloneqq\formel{conn}\wedge\forall x\;\formel{even}\big(\# x'\;E(x,x')\big),
  \]
  where $\formel{conn}$ is the sentence from Example~\ref{exa:conn}.
\end{example}

We like to think of definitions in the logic $\FPC$ in an
algorithmic way, where formulas are ``programs'' computing an
input-output relation. Rather than writing out syntactical details, we
describe these programs in a high-level form, as we would do with any
type of
algorithms and leave the ``$\FPC$-implementation'' to the reader. A
thorough technical treatment of the issues involved in this can be found
in \cite[Chapters~2~and~3]{gro17}.

We will make assertions about the existence of
$\FPC$-formulas, or ``programs'', with a specified input-output behaviour.
In the most basic setting of an $\FPC$-sentence like
$\formel{eulerian}$, the input is a structure (a graph) and the output
a Boolean value. For a formula like $\formel{even}(y)$ the input is a
structure $A$ and the output is a subset of $N(A)$. But the input can be
more complicated. For example, the input may be a triple $(G,U,d)$,
where $G$ is a graph, $U\subseteq V(G)$, $d\in N(G)$,
and our task is to write an $\FPC$-formula that computes the set of
all vertices of $G$ that have distance at most $d$ from a vertex in
$U$. Formally, this means that we have to write an $\FPC$-formula
$\phi(X,y,x)$ such that for all graphs $G$, all subsets $U\subseteq
V(G)$, and all $d\in N(G)$ we have $G\models\phi(U,d,v)$ if and only if
the distance of $v$ to $U$ in $G$ is at most $d$. An even more
complicated type of assertion we will frequently see is of the
following form: \emph{given a tuple $(G,U,v,H)$, where $G$ is a graph,
  $U\subseteq V(G)$, $v\in V(G)\setminus U$, and $H$ is
a graph with vertex set $V(H)\subseteq N(G)$, in $\FPC$ we can
decide if the connected component of $G\setminus U$ that
contains $v$ is isomorphic to $H$.} Here our task is to define an
$\FPC$-formula $\phi(X,x,Y,Z)$, where $X$ is a unary relation ranging
over the vertex sort, $x$ is a vertex variable, and $Y$ and $Z$ are a unary and
a binary relation symbol both ranging over the number sort, such that
for all graphs $G$, $U\subseteq V(G)$, $v\in V(G)\setminus U$,
$P\subseteq N(G)$, $Q\subseteq N(G)^2$ we have $G\models
\phi(U,v,P,Q)$ if and only if the connected component of $G\setminus U$ that
contains $v$ is isomorphic to the graph $H$ with $V(H)=P$ and
$E(H)=Q$.
Actually, it is known that such a formula can not exist (see \cite{caifurimm92}).
However, in this work we will only require such formulas for specific graph classes $\mathcal{C}$, that is, when the input is restricted to graphs $G \in \mathcal{C}$.

\begin{lemma}[cf.\ \cite{Otto2017}]
 \label{la:wl-in-ifpc}
 Let $\mathcal{C}$ be a hereditary graph class (i.e., $\mathcal{C}$ is closed under induced subgraphs) such that the $k$-dimensional Weisfeiler-Leman algorithm identifies all graphs $G \in \mathcal{C}$ for some constant number $k$.
 Then there is an $\FPC$-formula $\phi(X,Y,Z)$, where $X$ is a unary relation ranging
 over the vertex sort and $Y$ and $Z$ are a unary and a binary relation symbol both ranging over the number sort, such that
 for all graphs $G \in \mathcal{C}$, $U\subseteq V(G)$, $P\subseteq N(G)$, $Q\subseteq N(G)^2$ we have $G\models
 \phi(U,P,Q)$ if any only if $G[U] \cong (P,Q)$.
\end{lemma}

In a setting where we have several input objects, such as the tuple
$(G,U,v,H)$ above, we always have one \emph{main input structure},
which will be listed first. In the example $(G,U,v,H)$, this is the
graph $G$. All other objects are defined relative to this main
structure and its numerical extension. In the example, $U$ is a
subset of $V(G)$, $v$ an element of $V(G)$, and $H$ a structure with
universe $N(G)$. Sometimes, we will have to deal with whole families
of structures. They will always be indexed by tuples of elements of
the main structure. For example, we may be given a pair
$\big(A,(H_{(v,p)})_{(v,p)\in V(A)\times N(A)}\big)$ where the
$H_{(v,p)}$ are graphs with universe $V(H_{(v,p)})\subseteq
N(A)$. Formally, we can represent such a family by the ternary
relation
$R=\{(v,p,q)\in V(A)\times N(A)\times N(A) \mid q\in V(H_{(v,p)})\}$
and the quaternary relation
$S=\{(v,p,q,q')\in V(A)\times N(A)\times N(A) \times N(A)\mid (q,q')\in
E(H_{(v,p)})\}$.

\subsubsection*{Definable Canonisation}

Recall from the introduction that a logic captures polynomial time on a class $\CC$ of
structures if each polynomial-time decidable property of structures
in $\CC$ is expressible by a sentence of the logic. By the
Immerman-Vardi Theorem~\cite{imm87,var82}, $\IFP$ captures
polynomial time on the class of all ordered
structures.\footnote{Originally, the Immerman-Vardi Theorem states
  that \emph{least fixed-point logic} $\LFP$ captures polynomial time
  on the class of all ordered structures. However, it is known that
  $\LFP$ and $\IFP$ have the same expressive power
  \cite{gurshe86,kre04}.} 
A straightforward
way of applying this theorem to a class $\CC$ of unordered structures is to define a linear order on this class: if there is a formula
$\formel{ord}(x,y)$ of the
logic $\IFP$ that defines a linear order on all structures in $\CC$,
then $\IFP$ still captures polynomial time on $\CC$. Unfortunately,
this observation is rarely applicable, because usually it is
impossible to define linear orders. For example, it is impossible to
define a linear order on a structure that has a nontrivial
automorphism.

A much more powerful idea, going back to \cite{immlan90,Otto2017} and known as \emph{definable canonisation}, is to
define an \emph{ordered copy} of the input structure. To implement
this idea, $\FPC$ is particularly well-suited, because we can take
the numerical part $N(A)$ of a structure $A$ as the universe of the
ordered copy of $A$. Technically, definable canonisation is
based on syntactical interpretations (called transductions in
\cite{gro17}). Instead of introducing the unwieldy machinery of
syntactical interpretations in full generality, we just focus on a
special case that suffices for our purposes. Suppose, we have a
structure $A$ of
vocabulary $\tau$. To
define an ordered copy of $A$, we need a formula
$\phi_R(y_1,\ldots,y_{k})$ with free number variables $y_i$ for
every $k$-ary $R\in\tau$. A family $\Phi=(\phi_R(\bar y)\mid R\in\tau)$ of such formulas
defines a structure $A^\Phi$ with universe
$V(A^\Phi)\coloneqq\{0,\ldots,|V(A)|-1\}\subseteq N(A)$ and relations
$R(A^\Phi)\coloneqq\big\{(p_1,\ldots,p_{k})\in V(A^\Phi)^{k} \bigmid
A\models\phi_R(p_1,\ldots,p_{k})\big\}$. We say that $\Phi$ defines an
\emph{ordered copy} of $A$ if $A^\Phi\cong A$. Observe that if $\Phi$
defines an ordered copy of $A$, then this ordered copy is
\emph{canonical} in the sense that for all $B\cong A$ it holds that
$B^\Phi=A^\Phi$, because we have $N(B)=N(A)=\{0,\ldots,n\}$ for
$n=|V(A)|=|V(B)|$, and definitions in the logic $\FPC$ are
isomorphism-invariant. We say that a class $\CC$ of $\tau$-structures
\emph{admits $\FPC$-definable canonisation} if there is a
family $\Phi=(\phi_R(\bar y)\mid R\in\tau)$ of $\FPC$-formulas such that for
all $A\in\CC$ it holds that $A^{\Phi}\cong A$. The following lemma is a
direct consequence of the Immerman-Vardi Theorem (for a proof, see
\cite[Lemma~3.3.8]{gro17}).

\begin{lemma}
 \label{la:capture-p-from-definable-canon}
 Let $\CC$ be a class of $\tau$-structures that admits $\FPC$-definable
 canonisation. Then $\FPC$ captures polynomial time on $\CC$.
\end{lemma}

Sometimes, we need to define ordered copies of substructures of a
structure. To define an ordered copy of a substructure of a
$\tau$-structure we use a family $\Psi$ of formulas that in addition
to formulas $\psi_R(\bar y)$ for the relations contains a formula
$\psi_V(y)$ that specifies the universe of the ordered copy. Given a
pair $(A,B)$, where $A$ is a $\tau$-structure and $B\subseteq A$ a substructure, such
a family $\Psi$ defines a structure $B'$ with universe $V(B')\coloneqq\{p\in
N(A)\mid A\models\psi_V(p)\}$ and relations
$R(B')\coloneqq\{(p_1,\ldots,p_{k})\in V(B')^{k} \mid
A\models\psi_R(p_1,\ldots,p_{k})\}$. If $B'\cong B$, we say that
$\Psi$ \emph{defines an ordered copy of $B$ in $A$}.

We will also see more complicated assertions such as the following:
\emph{given a tuple $(G,U,v)$, where $G$ is a graph,
  $U\subseteq V(G)$, $v\in V(G)\setminus U$, in $\FPC$ we can compute
  an ordered copy of the connected component of $G\setminus U$ that
contains $v$.} This means that we can construct $\FPC$-formulas
$\psi_V(X,x,y)$ and $\psi_E(X,x,y_1,y_2)$ such that for all $G$,
$U\subseteq V(G)$, and $v\in V(G)\setminus U$, the graph with universe $V'\coloneqq\{p\in
N(G)\mid G\models\psi_V(U,v,p)\}$ and edge relation
$E'\coloneqq\{(p_1,p_2)\in (V')^{2} \mid
G\models\psi_E(U,v,p_1,p_2)\}$ is isomorphic to the connected
component of $v$ in $G\setminus U$.

We will routinely have to compare ordered copies of substructures of
our input graphs. To do this, we define a \emph{lexicographical order}
on $\tau$-structures whose universe is an initial segment of the nonnegative
integers. First, we fix a linear order of the
relation symbols in $\tau$. Say, $\tau=\{R_1,\ldots,R_\ell\}$ and we
order the $R_i$ by their indices. Now let $A,B$ be two
$\tau$-structures such that $V(A)=\{0,\ldots,n_A-1\}$ and
$V(B)=\{0,\ldots,n_B-1\}$. Structure $A$ is \emph{lexicographically
  smaller than or equal to} structure $B$ (we write $A\le_{\mathsf{lex}} B$)
if either $A=B$, or $A\neq B$ and $n_A=n_B$ and for the least $i\in[\ell]$
such that $R_i(A)\neq R_i(B)$ the lexicographically first tuple $\bar
p$ in the symmetric difference of $R_i(A)$ and $R_i(B)$ is contained
in $R_i(B)$, or $n_A<n_B$.

\subsection{Definable Canonisation of Graphs of Bounded Rank Width}

Recall that our goal is to prove that \FPC\ captures \PTIME\ on the
class of graphs of rank width at most $k$. Towards this end we prove
the following theorem. 

\begin{theorem}\label{theo:defcan}
  For every $k\ge 1$, the class of all graphs of rank width at most
  $k$ admits \FPC-definable canonisation.
\end{theorem}

Observe that, combined with Lemma~\ref{la:capture-p-from-definable-canon}, this theorem implies Theorem~\ref{theo:2}.

The rest of this section is devoted to a proof of
Theorem~\ref{theo:defcan}. Let us fix $k\ge 1$.
Our strategy to define an ordered copy of a graph $G$ of rank width
at most $k$  is similar to the proof strategy for showing that the Weisfeiler-Leman algorithm identifies such a graph.
For ordered split pairs $(\bar a,\bar b)$, flip functions $f$, and
components $C$ of the flipped graph we
recursively define an ordered copy of the induced subgraph $(G[C \cup \bar a \cup \bar b],\bar a,\bar b)$.
The first hurdle towards implementing this strategy is that we need to have explicit access to the flip function
(this is different from the previous section where we only needed the existence of such a function in order to describe a strategy for Spoiler).
However, we can not simply list all of the flip functions as there may
be exponentially many.
We remedy this by altering the definition of a flip so that, for every
fixed $k$, there is only a polynomial number of flips.

Throughout this section let $k \geq 1$ be a fixed natural number.
Let $G = (V,E)$ be a graph of rank width at most $k$ and let $n \coloneqq |V|$ denote the number of vertices of $G$.
In this section an \emph{ordered split pair of order at most $k$} is simply a pair $(\bar a,\bar b)$ where $\bar a,\bar b \in V^{\leq k}$.
For $v,w \in V$ we say that $v \approx_{(\bar a,\bar b)} w$ if $N(v) \cap (\bar a,\bar b) = N(w) \cap (\bar a,\bar b)$.
Clearly, $\approx_{(\bar a,\bar b)}$ defines an equivalence relation on $V$.
For tuples $\bar a,\bar b \in V^{\leq k}$ we denote by $2^{\bar a \cup \bar b}$ the set of all subsets of $\bar a \cup \bar b \subseteq V$ where we interpret the tuples $\bar a$ and $\bar b$ as subsets of $V$.
A \emph{flip extension} of an ordered split pair $(\bar a,\bar b)$ is a tuple
\[
  \bar s = \left(\bar a, \bar b, f \colon \left( 2^{\bar a \cup \bar
      b}\right)^{2} \rightarrow [n] \cup \{\perp\}\right)
\]
such that for all $M,N \in 2^{\bar a \cup \bar b}$ with $M \neq N$, either $f(M,N) = \perp$ or $f(N,M) = \perp$.
For $v,w \in V$ we say that $v \approx_{\bar s} w$ if $v \approx_{(\bar a,\bar b)} w$.
We denote by $[v]_{\approx_{\bar s}}$ the equivalence class of $v$ with respect to $\approx_{\bar s}$.
Moreover, we define the graph $G^{\bar s} = (V,E^{\bar s},\bar a,\bar b)$
where
\begin{align*}
 E^{\bar s} \coloneqq \;\;\; &\left\{vw \in E \mid f(N(v) \cap (\bar a,\bar b),N(w) \cap (\bar a,\bar b)) = d \in [n] \wedge |N(v) \cap [w]_{\approx_{\bar s}}| < d\right\}\\
                      \cup\; &\left\{vw \notin E \mid f(N(v) \cap (\bar a,\bar b),N(w) \cap (\bar a,\bar b)) = d \in [n] \wedge |N(v) \cap [w]_{\approx_{\bar s}}| \geq d\right\}.
\end{align*}
Finally we let $\comp(G,\bar s) \subseteq 2^{V}$ be the set of vertex sets of the connected components of the graph $G^{\bar s}$
and for $v \in V$ we define $\comp(G,\bar s, v)$ to be the unique $C \in \comp(G,\bar s)$ such that $v \in C$.

The following lemma is similar in nature to Lemma \ref{la:find-flip-function}.

\begin{lemma}
 \label{la:find-flip-extension}
 Let $G$ be a graph and let $X \subseteq V(G)$. Furthermore, let $(\bar a,\bar b)$ be an ordered split pair for $X$.
 Then there is a flip extension $\bar s = (\bar a,\bar b, f)$ such that $C \subseteq X$ or $C \subseteq \overline{X}$ for every $C \in \comp(G,\bar s)$.
\end{lemma}

The proof strategy for the lemma is similar to the proof of Lemma \ref{la:find-flip-function}.
Before giving the details, let us describe the main difference between the two proofs that also motivates our definition of a flip extension.
Let $v,w \in V(G)$ and consider the sets $\widehat{P} \coloneqq [v]_{\approx_{\bar s}}$ and $\widehat{Q} \coloneqq [w]_{\approx_{\bar s}}$.
As in the proof of Lemma \ref{la:find-flip-function}, the bipartite graph between $\widehat{P} \cap X$ and $\widehat{Q} \cap \overline{X}$ is either empty or complete (using Lemma \ref{la:same-neighbors-are-equivalent}).
In the latter case, we again wish to flip the edges between $\widehat{P} \cap X$ and $\widehat{Q} \cap \overline{X}$ in order to disconnect $X$ from $\overline{X}$.
However, other than in the proof of Lemma \ref{la:find-flip-function}, the bipartite graph between $\widehat{Q} \cap X$ and $\widehat{P} \cap \overline{X}$ may now be empty.
To handle this particular case, the main idea is to use a degree-threshold $d$ for one of the two sides of the bipartite graph between $\widehat{P}$ and $\widehat{Q}$ to identify those vertices lying in $X$.
In turn, this allows to identify the pairs $(v',w') \in \widehat{P} \times \widehat{Q}$ for which the edge relation needs to be flipped.

\begin{proof}
 Let $M,N \subseteq \bar a \cup \bar b$, let $c(M) \coloneqq \{v \in V(G) \mid N(v) \cap (\bar a,\bar b) = M\}$ and similarly $c(N) \coloneqq \{v \in V(G) \mid N(v) \cap (\bar a,\bar b) = N\}$.
 Let $P \coloneqq c(M) \cap X$, $\overline{P} \coloneqq c(M) \cap \overline{X}$, $Q \coloneqq c(N) \cap X$ and $\overline{Q} \coloneqq c(N) \cap \overline{X}$.
 
 We need to define $f$ in such a way such that $(P \times \overline{Q}) \cap E(G^{\bar s}) = \emptyset$ and $(Q \times \overline{P}) \cap E(G^{\bar s}) = \emptyset$.
 If $(P \times \overline{Q}) \cap E(G) = \emptyset$ and $(Q \times \overline{P}) \cap E(G) = \emptyset$ then we can simply set $f(M,N) = n$ and $f(N,M) = \perp$.
 So assume one of the two sets is non-empty.
 Without loss of generality suppose $v' \in Q$ and $w' \in \overline{P}$ such that $v'w' \in E(G)$.
 \begin{claim}
  \label{claim:edge-2}
  $Q \times \overline{P} \subseteq E(G)$.
 \end{claim}
 \begin{claimproof}
  Let $y \in Q$ and $z \in \overline{P}$.
  Then $v' \approx_X y$ and $w' \approx_{\overline{X}} z$ by Lemma \ref{la:same-neighbors-are-equivalent}.
  Hence,
  \[v'w' \in E(G) \;\;\Leftrightarrow\;\; v'z \in E(G) \;\;\Leftrightarrow\;\; yz \in E(G).\qedhere\]
 \end{claimproof}
 If $(P \times \overline{Q}) \subseteq E(G)$ and $(Q \times \overline{P}) \subseteq E(G)$ then we can simply set $f(M,N) = 1$ and $f(N,M) = \perp$.
 Hence, assume there are $v \in P$ and $w \in \overline{Q}$ such that $vw \notin E(G)$.
 \begin{claim}
  \label{claim:non-edge-2}
  $(P \times \overline{Q}) \cap E(G) = \emptyset$.
 \end{claim}
 \begin{claimproof}
  Let $y \in P$ and $z \in \overline{Q}$.
  Then $v \approx_X y$ and $w \approx_{\overline{X}} z$ by Lemma \ref{la:same-neighbors-are-equivalent}.
  Hence,
  \[vw \in E(G) \;\;\Leftrightarrow\;\; vz \in E(G) \;\;\Leftrightarrow\;\; yz \in E(G).\qedhere\]
 \end{claimproof}
 \begin{claim}
  The following inequalities hold:
  \begin{enumerate}[label = (\alph*)]
   \item $|N(v'') \cap c(N)| \leq |N(w'') \cap c(N)|$ for all $v'' \in P$ and $w'' \in \overline{P}$.
   \item $|N(v''') \cap c(M)| \geq |N(w''') \cap c(M)|$ for all $v''' \in Q$ and $w''' \in \overline{Q}$.
  \end{enumerate}
  Moreover, one of the two inequalities is strict.
 \end{claim}
 \begin{claimproof}
  We have  $N(v'') \cap c(N) \subseteq Q \subseteq N(w'') \cap c(N)$ and $N(v''') \cap c(M) \supseteq \overline{P} \supseteq N(w''') \cap c(M)$ by Claim \ref{claim:edge-2} and \ref{claim:non-edge-2}.
  This proves the inequalities.
  To argue that one of the two inequalities is strict suppose that $|N(v'') \cap c(N)| = |N(w'') \cap c(N)|$ for all $v'' \in P$ and $w'' \in \overline{P}$.
  Then $N(v'') \cap c(N) = Q$ for all $v'' \in P$.
  But now $N(v''') \cap c(M) \supsetneq \overline{P}$ for all $v''' \in Q$ since $v \in N(v''')$.
  Thus, $|N(v''') \cap c(M)| > |N(w''') \cap c(M)|$ for all $v''' \in Q$ and $w''' \in \overline{Q}$.
 \end{claimproof}
 Without loss of generality assume that the first inequality of the previous claim is strict.
 Let $d \coloneqq \min_{w'' \in \overline{P}} |N(w'') \cap c(N)|$.
 We set $f(M,N) \coloneqq d$ and $f(N,M) \coloneqq \perp$.
 Then $|N(v'') \cap c(N)| < d$ for all $v'' \in P$ and hence, $(P \times \overline{Q}) \cap E(G^{\bar s}) = \emptyset$ using Claim \ref{claim:non-edge-2}.
 Also $|N(w'') \cap c(N)| \geq d$ for all $w'' \in \overline{P}$
 and thus, $(Q \times \overline{P}) \cap E(G^{\bar s}) = \emptyset$ using Claim \ref{claim:edge-2}.
\end{proof}

\begin{lemma}
 \label{la:nice-tuples-down}
 Let $G$ be a graph, $X_1 \subseteq X \subseteq V(G)$, $(\bar a,\bar b)$ an ordered split pair for $X$
 and $(\bar a_1,\bar b_1)$ an ordered split pair for $X_1$ such that $X_1 \cap \bar a \subseteq \bar a_1$.
 Let $v,w \in X_1$ such that $v \approx_{(\bar a_1, \bar b_1)} w$.
 Then $v \approx_{(\bar a,\bar b)} w$.
\end{lemma}

\begin{proof}
 Let $v,w \in X_1$ such that $v \approx_{(\bar a_1, \bar b_1)} w$.
 Then $v \approx_{X_1} w$ by Lemma \ref{la:same-neighbors-are-equivalent}.
 So $N(v) \cap \overline{X_1} = N(w) \cap \overline{X_1}$ and $N(v) \cap (\bar a_1,\bar b_1) = N(w) \cap (\bar a_1,\bar b_1)$.
 Since $X_1 \cap \bar a \subseteq \bar a_1$ and $X_1 \cap \bar b = \emptyset$ this implies $v \approx_{(\bar a, \bar b)} w$.
\end{proof}

\begin{lemma}
 \label{la:nice-tuples-sib}
 Let $G$ be a graph, $X_1, X_2 \subseteq V(G)$ such that $X_1 \cap X_2 = \emptyset$, $(\bar a_1,\bar b_1)$ an ordered split pair for $X_1$
 and $(\bar a_2,\bar b_2)$ an ordered split pair for $X_2$ such that $X_1 \cap \bar b_2 \subseteq \bar a_1$.
 Let $v,w \in X_1$ such that $v \approx_{(\bar a_1, \bar b_1)} w$.
 Then $v \approx_{(\bar a_2, \bar b_2)} w$.
\end{lemma}

\begin{proof}
 Let $v,w \in X_1$ such that $v \approx_{(\bar a_1, \bar b_1)} w$.
 Then $v \approx_{X_1} w$ by Lemma \ref{la:same-neighbors-are-equivalent}.
 So $N(v) \cap \overline{X_1} = N(w) \cap \overline{X_1}$ and $N(v) \cap (\bar a_1,\bar b_1) = N(w) \cap (\bar a_1,\bar b_1)$.
 Since $X_1 \cap \bar a_2 = \emptyset$ and $X_1 \cap \bar b_2 \subseteq \bar a_1$ this implies $v \approx_{(\bar a_2, \bar b_2)} w$.
\end{proof}

\begin{definition}
 \label{def:anchored-copy}
 Let $G = (V,E)$ be a graph.
 Let $\bar s = (\bar a,\bar b,f)$ be a flip extension and let $C \subseteq V(G)$.
 An \emph{$\bar s$-anchored (ordered) copy of $C$} is a tuple $C_\anc = (V_\anc, E_\anc,\bar a_\anc,\bar b_\anc,\eta)$
 such that
 \begin{enumerate}[label = (\roman*)]
  \item\label{it:anc-cop1} $V_\anc$ is an initial segment of the nonnegative integers,
  \item\label{it:anc-cop2} $\eta\colon V_\anc \rightarrow (2^{\bar a \cup \bar b} \rightarrow \mathbb{N})$, and
  \item\label{it:anc-cop3} there is an isomorphism $\sigma\colon (G[C \cup \bar a \cup \bar b],\bar a,\bar b) \cong (V_\anc, E_\anc,\bar a_\anc,\bar b_\anc)$
   such that $(\eta(i))(M) = |N(\sigma^{-1}(i)) \cap \{w \in V(G) \mid
   N(w) \cap (\bar a,\bar b) = M\}|$ for all $i \in V_\anc$.
 \end{enumerate}
\end{definition}

In the following, this definition is typically applied to sets
$C \in \comp(G,\bar s)$.  More precisely, our aim is to define
$\bar s$-anchored copies of $C$ for all $C \in \comp(G,\bar s)$ and
all suitable flip extensions $\bar s$ in $\FPC$ in an inductive
fashion.  An important feature of an $\bar s$-anchored copy of $C$ is
that, in addition to being an ordered copy of $C$, it records some
``context information'' on how the set $C$ is connected to the rest of
the graph. This is the purpose of the function $\eta$.  The context
information will play a vital role in the proofs since it gives the
relevant information to perform flips also in the $\bar s$-anchored
copy of $C$.

Let us start by discussing how to represent the relevant objects in
the logical framework.  Since $k$ is a constant we can view a flip
extension as a tuple
$\bar s \in V(G)^{2k} \times \{0,\dots,n\}^{2^{2k}}$ of fixed length.
The first $2k$ components represent the split pair $(\bar a,\bar b)$
and the function
$f\colon 2^{\bar a \cup \bar b} \rightarrow [n] \cup \{\perp\}$ can be
seen as a tuple in $\{0,\dots,n\}^{2^{2k}}$ where a $0$-entry is
interpreted as $\perp$.  Similar to the previous section, components
(with respect to some flip extension) are represented by a single
vertex from that component.  To be more precise, a set
$C \in \comp(G,\bar s)$ is represented by $\bar s$ and some $v \in C$.
Observe that there is an $\FPC$-formula $\varphi(\bar x,y,z)$ such
that $G \models \varphi(\bar s,v,w)$ if and only if
$w \in \comp(G,\bar s ,v) = C$ (see Example \ref{exa:conn} for how to
define reachability in $\FPC$).

To represent an $\bar s$-anchored copy $C_\anc = (V_\anc, E_\anc,\bar a_\anc,\bar
b_\anc,\eta)$ of a set $C \subseteq V(G)$, we represent the function
$\eta$ by a relation $P_\eta \subseteq N(G)^{1+2^{2k}}$ containing elements $(p,\bar p_\eta)$ for every $p \in V_\anc$.
The tuple $\bar p_\eta\in N(G)^{2^{2k}}$ represents the function $\eta(p)$ and has an entry for each
subset $M\subseteq \bar a\cup\bar b$ which specifies $(\eta(p))(M)$.
Typically, we will denote the set $V_\anc\subseteq N(G)$ by $P_V$, the relation $E_\anc\subseteq
N(G)^2$ by $P_E$, the tuple $\bar a_\anc$ by $\bar p_a$, the tuple
$\bar b_\anc$ by $\bar p_b$, and the relation representing $\eta$ by
$P_\eta$. Slightly abusing notation, we will write $C_\anc=(P_V,P_E,\bar p_a,\bar
p_b,\bar P_\eta)$.
We can define such an anchored copy by \FPC-formulas
$\psi_V(z),\psi_E(z_1,z_2),\psi_{\bar a}(\bar z),\psi_{\bar b}(\bar
z),\psi_\eta(z,\bar z')$, where $z,z_1,z_2$ are number variables and
$\bar z,\bar z'$ are tuples of number variables of lengths $k,2^{2k}$, respectively. 

We now start by constructing various $\FPC$ formulas.
They will form the basic building blocks of the formulas defining $\bar s$-anchored copies of $C$ for sets $C \in \comp(G,\bar s)$.

\begin{lemma}
 \label{lem:test-anchored}
 Let $G$ be a graph of rank width at most $k$.
 Also let $\bar s$ be a flip extension, $D \in \comp(G,\bar s)$,
 and $D^{*}=(P_V,P_E,\bar p_a,\bar p_b,P_\eta)$ where
 $P_V\subseteq N(G)$, $P_E\subseteq N(G)^2$, $\bar p_a,\bar p_b \in N(G)^k$ and $P_\eta \subseteq N(G)^{1+2^{2k}}$. 

 There is an $\FPC$-sentence that, given access to the objects $(G,\bar s,D,D^{*})$, decides if $D^{*}$ is an $\bar s$-anchored copy of $D$.
\end{lemma}

\begin{proof}
 The proof of this lemma uses similar arguments as the proof of Lemma \ref{la:wl-in-ifpc}.
 Suppose $\bar s = (\bar a,\bar b, f)$ and let $R_\eta = \{(v,\bar v_\eta) \mid v \in V\}$
 where $\bar v_\eta \in N(G)^{2^{2k}}$ has an entry for each
 subset $M\subseteq \bar a\cup\bar b$ which specifies $|N(v) \cap \{w \in V(G) \mid N(w) \cap (\bar a,\bar b) = M\}|$.
 
 We need to check whether $(G[D],\bar a,\bar b,R_\eta) \cong D^{*}$.
 This can be achieved by implementing the Weisfeiler-Leman algorithm within fixed-point logic with counting.
 Here, the algorithm additionally needs to take the vertex-colouring into account that is given by $(\bar a,\bar b,R_\eta)$ and $(\bar p_a,\bar p_b,P_\eta)$, respectively
\end{proof}

Given an $\bar s$-anchored copy $C_\anc$ of a set $C \subseteq V(G)$, we regularly need to associate elements from $C_\anc$ with their corresponding element in $C$ (via a bijection $\sigma$ according to Definition \ref{def:anchored-copy}\ref{it:anc-cop3}).
Unfortunately, it is not possible to compute $\sigma$ directly within \FPC\ as it may not be unique.
The next lemma serves as tool that allows us to answer all necessary queries without directly accessing $\sigma$.

\begin{lemma}
 \label{la:recover-info-down}
 Let $G$ be a graph, $X_1 \subseteq X \subseteq V(G)$, $(\bar a,\bar b)$ an ordered split pair for $X$ and $(\bar a_1,\bar b_1)$ an ordered split pair for $X_1$.
 Let $\bar s = (\bar a,\bar b,f)$ and $\bar s_1 = (\bar a_1,\bar b_1,f_1)$ be flip extensions.
 Moreover, suppose that $X_1 \cap \bar a \subseteq \bar a_1$ and let $D \in \comp(G,\bar s_1)$ such that $D \subseteq X_1$.
 Also let $D_\anc$ be an $\bar s_1$-anchored copy of $D$ and let
 $\sigma$ be an isomorphism according to Definition~\ref{def:anchored-copy}\ref{it:anc-cop3}.
 
 Given access to the objects $(G,\bar s,\bar s_1,D,D_\anc)$, the following queries can be defined using \FPC-formulas:
 \begin{enumerate}[label = (\arabic*)]
  \item\label{item:down-1} given $p \in V(D_\anc)$, determine $N(v) \cap (\bar a,\bar b)$ where $v = \sigma^{-1}(p)$,
  \item\label{item:down-2} given $p \in V(D_\anc)$ and $v' \in V(G)$, determine whether $v \approx_{\bar s} v'$  where $v = \sigma^{-1}(p)$,
  \item\label{item:down-3} given $p \in V(D_\anc)$ and $w \in V(G)$, determine $|N(v) \cap [w]_{\approx_{\bar s}}|$ where $v = \sigma^{-1}(p)$, and
  \item\label{item:down-4} given $p,q \in V(D_\anc)$, determine whether $vw \in E(G^{\bar s})$ where $v = \sigma^{-1}(p)$ and $w = \sigma^{-1}(q)$.
 \end{enumerate}
\end{lemma}

Just to be on the safe side, let us again explain the exact technical
meaning of the assertions of the lemma,  taking assertion
\ref{item:down-1} as an example.
We need to construct an $\FPC$-formula $\varphi(\bar x,\bar
x_1,y,Z_V,Z_E,\bar z_a,\bar z_b,Z_\eta,z,y')$, where $Z_V$ is a unary relation symbol of type
'number', $Z_E$ is a binary relation
symbol of type 'number$\times$number', $Z_\eta$ is a $(1+2^{2k})$-ary
relation symbol of type $\text{('number')}^{(1+2^{2k})}$, $y,y'$ are vertex variables, $z$
is a number variable, $\bar z_a$, $\bar z_b$ are $k$-tuples of number variables, and $\bar x,\bar x_1$ are tuples of individual
variables of the type appropriate for representing flip
extensions. The formula
is supposed to have the following meaning. Suppose that
$D_{\anc}=(P_V,P_E,\bar p_a,\bar p_b,P_\eta)$. Then for
all $u\in D$, $p\in P_V$, and $w\in V(G)$,
\[
  G \models \varphi(\bar s,\bar s_1,u,P_V,P_E,\bar p_a,\bar p_b,P_\eta,p,w)
\]
if and only if $w \in N(v) \cap (\bar a,\bar b)$.
We should think of the relations $P_V,P_E,P_\eta$, which determine the core
of the structure
$D_{\anc}$, as being defined earlier in some inductive process. Note
that we do not specify $D$ explicitly in the definition, but only
implicitly by giving the flip extension $\bar s_1$ and the vertex $u$.

\begin{proof}[Proof of Lemma~\ref{la:recover-info-down}]
  Let $p \in V(D_\anc)$ and let $v = \sigma^{-1}(p)$.  Given $p$, one
  can clearly find some $v'' \in D$ such that
  $v \approx_{\bar s_1} v''$.  By Lemma \ref{la:nice-tuples-down} we
  conclude that $v \approx_{\bar s} v''$.  Using $v''$ instead of $v$
  one can already solve Tasks \ref{item:down-1} and \ref{item:down-2}.
 
 So consider Task \ref{item:down-3}.
 First note that we can easily define the set $[w]_{\approx_{\bar s}}$
 given $w$, and $\bar a$, $\bar b$, which are both contained in $\bar s$.
 Now let
 \[Y_1 = \bigcup_{u \in D\colon  u \approx_{\bar s_1} v''} N(u) \bigtriangleup N(v'')\]
 (here $P \bigtriangleup Q$ denotes the symmetric difference between the two sets $P$ and $Q$).
 Then $Y_1 \subseteq X_1$ and $N(v) \cap \overline{Y_1} = N(v'') \cap \overline{Y_1}$ by Lemma \ref{la:same-neighbors-are-equivalent}.
 Now we compute
 \begin{align*}
  |N(v) \cap [w]_{\approx_{\bar s}}| = &\,|N(v) \cap [w]_{\approx_{\bar s}} \cap Y_1| + |N(v) \cap [w]_{\approx_{\bar s}} \cap \overline{Y_1}|\\
                                     = &\,|N(v) \cap [w]_{\approx_{\bar s}} \cap Y_1| + |N(v'') \cap [w]_{\approx_{\bar s}} \cap \overline{Y_1}|.
 \end{align*}
 The second term can be computed easily,
 so we only have to determine $|N(v) \cap [w]_{\approx_{\bar s}} \cap Y_1|$.
 For $M \subseteq (\bar a_1,\bar b_1)$ we define $c(M) \coloneqq \{u \in V(G) \mid N(u) \cap (\bar a_1,\bar b_1) = M\}$.
 We have that
 \[|N(v) \cap [w]_{\approx_{\bar s}} \cap Y_1| = \sum_{M \subseteq (\bar a_1,\bar b_1)} |N(v) \cap [w]_{\approx_{\bar s}} \cap Y_1 \cap c(M)|.\]
 Using the fact that $Y_1 \subseteq X_1$ and Lemma \ref{la:nice-tuples-down} we get that
 \begin{equation}
  \label{eq:recover-infos}
  [w]_{\approx_{\bar s}} \cap c(M) \cap Y_1 \neq \emptyset \;\;\;\Rightarrow\;\;\; [w]_{\approx_{\bar s}} \cap c(M) \cap Y_1 = c(M) \cap Y_1.
 \end{equation}
 Hence,
 \[|N(v) \cap [w]_{\approx_{\bar s}} \cap Y_1| = \sum_{\substack{M \subseteq (\bar a_1,\bar b_1)\colon\\ [w]_{\approx_{\bar s}} \cap c(M) \cap Y_1 \neq \emptyset}} |N(v) \cap Y_1 \cap c(M)|.\]
 But $|N(v) \cap Y_1 \cap c(M)| = |N(v) \cap c(M)| - |N(v) \cap c(M) \cap \overline{Y_1}| = |N(v) \cap c(M)| - |N(v'')  \cap c(M) \cap \overline{Y_1}|$.
 Recalling $|N(v) \cap c(M)| = (\eta(p))(M)$, the last term is easy to
 compute.
 
 Finally note that Task \ref{item:down-4} can be solved using the first three results in order to determine $N(v) \cap (\bar a,\bar b)$, $N(w) \cap (\bar a,\bar b)$ and $|N(v) \cap [w]_{\approx_{\bar s}}|$.
\end{proof}

\begin{remark}
 \label{remark:recover-infos}
 In later proofs we wish to apply the lemma to sets $D \in \comp(G,\bar s_1)$ without exactly knowing whether $D \subseteq X_1$.
 However, we do not need to know whether $D \subseteq X_1$ in order to determine whether the lemma is applicable.
 One of the crucial steps in the proof of the previous lemma is to obtain Equation \eqref{eq:recover-infos}.
 Indeed, besides being able to apply Lemma \ref{la:nice-tuples-down}, this is the only place where we need that $D \subseteq X_1$.
 Since we can easily check within $\FPC$ whether Equation \eqref{eq:recover-infos} holds,
 we can also find an $\FPC$-formula that checks whether the lemma can be applied given $(\bar s,\bar s_1,u)$ where $u \in D$ is arbitrary.
\end{remark}

We also need the following variant of the previous lemma.

\begin{lemma}
 \label{la:recover-info-sib}
 Let $G$ be a graph, $X_1,X_2 \subseteq V(G)$ such that $X_1 \cap X_2 = \emptyset$, $(\bar a_1,\bar b_1)$ an ordered split pair for $X_1$ and $(\bar a_2,\bar b_2)$ an ordered split pair for $X_2$.
 Let $\bar s_1 = (\bar a_1,\bar b_1,f_1)$ and $\bar s_2 = (\bar a_2,\bar b_2,f_2)$ be flip extensions.
 Moreover, suppose that $X_1 \cap \bar b_2 \subseteq \bar a_1$ and let $D \in \comp(G,\bar s_1)$ such that $D \subseteq X_1$.
 Also let $D_\anc$ be an $\bar s_1$-anchored copy of $D$ and let $\sigma$ denote any isomorphism according to Definition~\ref{def:anchored-copy}\ref{it:anc-cop3}. 
 
 Given access to the objects $(G,\bar s_1,\bar s_2,D,D_\anc)$, the following queries can be defined using \FPC-formulas:
 \begin{enumerate}[label = (\arabic*)]
  \item\label{item:sib-1} given $p \in V(D_\anc)$, determine $N(v) \cap (\bar a_2,\bar b_2)$ where $v = \sigma^{-1}(p)$,
  \item\label{item:sib-2} given $p \in V(D_\anc)$ and $v' \in V(G)$, determine whether $v \approx_{\bar s_2} v'$  where $v = \sigma^{-1}(p)$,
  \item\label{item:sib-3} given $p \in V(D_\anc)$ and $w \in V(G)$, determine $|N(v) \cap [w]_{\approx_{\bar s_2}}|$ where $v = \sigma^{-1}(p)$, and
  \item\label{item:sib-4} given $p,q \in V(D_\anc)$, determine whether $vw \in E(G^{\bar s_2})$ where $v = \sigma^{-1}(p)$ and $w = \sigma^{-1}(q)$.
 \end{enumerate}
\end{lemma}

\begin{proof}
  Analogous to the proof of Lemma~\ref{la:recover-info-down} using Lemma
  \ref{la:nice-tuples-sib} instead of Lemma \ref{la:nice-tuples-down}.
\end{proof}

Recall the definition of nice triples of ordered split pairs (see Definition \ref{def:nice-split-pairs}).

\begin{lemma}\label{lem:inductive-step}
 Let $G$ be a graph and $X,X_1,X_2 \subseteq V(G)$ such that $X = X_1 \uplus X_2$.
 Let $(\bar a,\bar b)$ be an ordered split pair for $X$ and let $(\bar
 a_i,\bar b_i)$ be ordered split pairs for $X_i$, $i \in \{1,2\}$,
 that are nice with respect to $(\bar a,\bar b)$.
 Moreover, let $\bar s = (\bar a,\bar b,f)$ and $\bar s_i = (\bar
 a_i,\bar b_i,f_i)$, $i \in \{1,2\}$, be flip extensions.
 Let $C \in \comp(G,\bar s)$ such that $C \subseteq X$, and let $\mathcal{D}_i \subseteq \comp(G,\bar s_i)$ such that
 \[X_i \subseteq \bigcup_{D \in \mathcal{D}_i} D.\]
 For every $D \in \mathcal{D}_i$ let $D_\anc$
 be an $\bar s_i$-anchored copy of $D$ and let
 $\mathcal{D}_i^{\anc}$ be the set of all of those copies for $i \in \{1,2\}$.
 
 Then in $\FPC$, given access to
 $(G,\bar s,\bar s_1,\bar s_2,C,\mathcal{D}_1,\mathcal{D}_1^{\anc},\mathcal{D}_2,\mathcal{D}_2^{\anc})$,
 we can define an $\bar s$-anchored copy $C_\anc$ of $C$.
\end{lemma}

Let us again discuss the precise meaning of this statement.
Specifically, we need to elaborate on how to represent the families of
sets $\mathcal{D}_i$ and the families of $\bar s_i$-anchored copies
$\mathcal{D}_i^{\anc}$. We index the two families by elements $u\in
V(G)$. For every $u\in V(G)$ we let $D_{i,u}\in\comp(G,\bar
s_i)$ be the unique component with $u\in D_{i,u}$. If $D_{i,u}\in\mathcal D_i$, we denote its anchored copy in
$\mathcal D_i^\anc$ by $D_{i,u}^\anc$. Note, however, that we do not
necessarily have $D_{i,u}\in\mathcal D_i$ for all $u$. The only
requirement is that $X_i$ is a subset of the union of all $D_{i,u}$ in
$\mathcal D$.
We represent the family $\mathcal D_i^\anc$ by a binary relation
$P_{i,V}\subseteq V(G)\times N(G)$, a ternary relation
$P_{i,E}\subseteq V(G)\times N(G)^2$, $(k+1)$-ary relations
$P_{i,\bar a},P_{i,\bar b}\subseteq V(G)\times N(G)^k$, and a $(2^{2k}+2)$-ary
relation $P_{i,\eta}$ such that if $u\in V(G)$ with $D_{i,u}\in\mathcal
D$ then $D_{i,u}^\anc=(V,E,\bar a,\bar b,\eta)$ where
\begin{itemize}
 \item $V=\{p\in N(G)\mid (u,p)\in P_{i,V}\}$;
 \item $E=\{(p,q)\in N(G)^2\mid (u,p,q)\in P_{i,E}\}$;
 \item $\bar a\in N(G)^k$ is the unique tuple with $(u,\bar a)\in P_{i,\bar a}$;    
 \item $\bar b\in N(G)^k$ is the unique tuple with $(u,\bar b)\in P_{i,\bar b}$;
 \item $\eta$ is represented by $P_\eta \subseteq N(G)^{1 + 2^{2k}}$ where
  $P_\eta = \{\bar p \in N(G)^{1+2^{2k}} \mid (u,\bar p)\in P_{i,\eta}\}$.
\end{itemize}
To define $C_\anc$, we need to construct $\FPC$-formulas $\phi_V$, $\phi_E$, $\phi_{\bar a}$, $\phi_{\bar b}$, $\phi_\eta$. They all have free variables $\bar x,\bar
x_1,\bar x_2$ for the flip extensions $\bar s,\bar s_1,\bar s_2$, a
free vertex variable $y$ for an element of the component $C$, and for
$i=1,2$ free relation variables $Z_{i,V},Z_{i,E},Z_{i,a},Z_{i,b},Z_{i,\eta}$
for the family $\mathcal D_i^\anc$. The family $\mathcal D_i$ is only
specified implicitly: $\mathcal D_i$ consists of all components
$D\in\comp(G,\bar s_i)$ such that a $u\in D$ appears as an index of a
structure in $\mathcal D_i^\anc$, that is, as the first component of a tuple in
the relations $Z_{i,V},Z_{i,E},Z_{i,a},Z_{i,b},Z_{i,\eta}$. In
addition, the formula $\phi_V$ has a free number variable $z$ for the
elements of $V(C_\anc)$. The formula $\phi_E$ has two free number variables $z_1,z_2$ for the
elements of $E(C_\anc)$. The formula $\phi_{\bar a},\phi_{\bar b}$ have a $k$-tuple of free
number variables $\bar z$ for $\bar a_\anc,\bar b_\anc$,
respectively. And finally, the formula $\phi_\eta$ has a $(1+2^{2k})$-tuple of free
number variables $\bar z$ to specify the function $\eta$ of $C_\anc$.

\begin{proof}[Proof of Lemma~\ref{lem:inductive-step}]
 First, we can assume without loss of generality that we can apply Lemma \ref{la:recover-info-down} and \ref{la:recover-info-sib} to every $D \in \mathcal{D}_i$ (by eliminating all sets that do not satisfy the requirements, see Remark \ref{remark:recover-infos}).
 
 Let $i \in \{1,2\}$.  Let $D \in \mathcal{D}_i$ and let
 $D_\anc \in \mathcal{D}_i^{\anc}$ be the $\bar s_i$-anchored copy of
 $D$.  We define $D_\anc^{\bar s} = (V(D_\anc),E^{\bar s})$ where
 $E^{\bar s} = \{ij \mid \sigma_D^{-1}(i)\sigma_D^{-1}(j) \in
 E(G^{\bar s})\}$ and $\sigma_D$ is an isomorphism to the
 $\bar s_i$-anchored copy according to Definition
 \ref{def:anchored-copy}\ref{it:anc-cop3}.  Note that we can define
 $D_\anc^{\bar s}$ in \FPC\ by Lemma
 \ref{la:recover-info-down}\ref{item:down-4}.  Now let
 $D' \subseteq D$ be a connected component of $(G^{\bar s})[D]$.  Then
 $D' \subseteq C$ or $D' \cap C = \emptyset$.  Also $\sigma_D(D')$ is
 a connected component of $D_\anc^{\bar s}$.  Since we can match the
 connected components of $(G^{\bar s})[D]$ to those of
 $D_\anc^{\bar s}$ by implementing the Weisfeiler-Leman algorithm in
 fixed-point logic with counting (see Lemma
 \ref{la:wl-in-ifpc}), we can define an
 $\bar s_i$-anchored copy $D_\anc^{C}$ of $D \cap C$.  Let
 $C_i = C \cap \bigcup_{D \in \mathcal{D}_i} D$.  By lexicographically
 ordering the $\bar s_i$-anchored copies of the sets $D \cap C$,
 $D \in \mathcal{D}_i$, we can also define an $\bar s_i$-anchored copy
 $C_i^{\anc}$ of $C_i$.  Note that we can figure out whether there is
 an edge between two vertices of different components by looking at
 the $\eta$-functions.  Indeed, in order to know whether there is an
 edge between $p \in D$ and $q \in D'$ we need to be able to compute,
 for $v = \sigma_D^{-1}(p)$ and $w = \sigma_{D'}^{-1}(q)$, the sets
 $N(v) \cap (\bar a_i,\bar b_i)$, $N(w) \cap (\bar a_i,\bar b_i)$ and
 the value $|N(v) \cap [w]_{\approx_{\bar s_i}}|$. Then the flip
 extension $\bar s_i$ tells us whether there is an edge or not.
 The neighbourhoods to $(\bar a_i,\bar b_i)$
 are directly given in the anchored copy and the number of neighbours
 in some given equivalence class is stored in the $\eta$-function.
 
 Essentially repeating this process applying Lemma \ref{la:recover-info-sib} instead of Lemma \ref{la:recover-info-down} we can also compute an $\bar s_2$-anchored copy $C_{2\setminus1}^{\anc}$ of the set $C_2 \setminus C_1$.
 
 Next, we can turn the $\bar s_1$-anchored copy $C_1^{\anc}$ of $C_1$ into an $\bar s$-anchored copy $C_1^{\anc,\bar s}$ of $C_1$ using Lemma \ref{la:recover-info-down}, Item \ref{item:down-1} and \ref{item:down-3}.
 Similarly, we can turn the $\bar s_2$-anchored copy $C_{2 \setminus 1}^{\anc}$ of $C_2\setminus C_1$ into an $\bar s$-anchored copy $C_{2\setminus 1}^{\anc,\bar s}$ of $C_1\setminus C_2$.
 
 Finally, in order to obtain $C_\anc$, we take the disjoint union of $C_1^{\anc,\bar s}$ and $C_{2\setminus 1}^{\anc,\bar s}$ (where $C_1^{\anc,\bar s}$ comes first, i.e.\ the corresponding vertices get smaller numbers assigned than the vertices from $C_{2\setminus 1}^{\anc,\bar s}$).
 It only remains to recover the edges between the two sides.
 But this can be done using Lemma \ref{la:recover-info-sib}, Item \ref{item:sib-1} and \ref{item:sib-3} reconstructing the information whether such a given edge is flipped by $\bar s_1$.
\end{proof}

With this, we are ready to prove our theorem.

\begin{proof}[Proof of Theorem~\ref{theo:defcan}]
  Let $G$ be a graph of rank width at most $k$. We will inductively
  construct $\bar s$-anchored copies $C_\anc$ for flip extensions
  $\bar s$ and components $C\in\comp(G,\bar s)$, using
  Lemma~\ref{lem:inductive-step} in the inductive step. The base step
  for components $C$ consisting of a single element will be easy.

  To describe the proof, we fix $G$, but of course the $\FPC$-formulas we
  shall construct will not depend on the specific graph $G$ and will
  work for every graph of rank width at most $k$.

  Note first that the set of all flip extensions of $G$, viewed as
  tuples in $V(G)^{2k}\times N(G)^{2^{2k}}$, is definable in $\FPC$:
  we only need to make sure that the part of the tuple in
  $N(G)^{2^{2k}}$ representing the flip function adheres to the
  simple conditions in the definition of a flip extension. Let
  $F(G)\subseteq V(G)^{2k}\times N(G)^{2^{2k}}$ be the set
  of all flip extensions.

  In our main induction, to be implemented by an \FPC-formula, we will
  define an increasing collection of anchored copies of components
  $C\in\comp(G,\bar s)$ for flip extensions $\bar s$.
  We shall simultaneously define five relations for every $\ell\ge 1$.
  Recall that for a flip extension $\bar s$ and $u \in V(G)$ we denote by $\comp(G,\bar s,u)$ the unique $C \in \comp(G,\bar s)$ such that $u \in C$.
  \begin{itemize}
  \item $R^{(\ell)}\subseteq F(G)\times V(G)$
    will consist of those tuples $(\bar s,u)$ such that we have
    already defined an anchored copy $C_{\anc}$ of the component $\comp(G,\bar s,u)$;
  \item $R^{(\ell)}_V \subseteq F(G)\times V(G)\times
    N(G)$ will consist of all tuples $(\bar s,u,p)$ such that $(\bar
    s,u)\in R^{(\ell)}$ and $p\in V(C_{\anc})$ for the anchored copy $C_\anc$
    of the component $\comp(G,\bar s,u)$;
  \item $R^{(\ell)}_E \subseteq F(G)\times V(G)\times
    N(G)^2$ will consist of all tuples $(\bar s,u,p_1,p_2)$ such that $(\bar
    s,u)\in R^{(\ell)}$ and $(p_1,p_2)\in E(C_{\anc})$ for the anchored copy $C_\anc$
    of the component $\comp(G,\bar s,u)$;
 \item $R^{(\ell)}_{\bar a}\subseteq F(G)\times V(G)\times
    N(G)^k$ will consist of all tuples $(\bar s,u,\bar p)$ such that $(\bar
    s,u)\in R^{(\ell)}$ and $\bar p=\bar a_{\anc}$ for the anchored copy $C_\anc$
    of the component $\comp(G,\bar s,u)$;
 \item $R^{(\ell)}_{\bar b}\subseteq F(G)\times V(G)\times
    N(G)^k$ will consist of all tuples $(\bar s,u,\bar p)$ such that $(\bar
    s,u)\in R^{(\ell)}$ and $\bar p=\bar b_{\anc}$ for the anchored copy $C_\anc$
    of the component $\comp(G,\bar s,u)$;
  \item $R^{(\ell)}_\eta\subseteq F(G)\times V(G)\times N(G)^{1+2^{2k}}$ will
    consist of all tuples $(\bar s,u,\bar p)$ such that
    $(\bar s,u)\in R^{(\ell)}$ and $\bar p \in R_\eta$ where $R_\eta$ represents the function $\eta$ of
    the anchored copy $C_\anc$ of the component $\comp(G,\bar s,u)$.
\end{itemize}
In the \FPC-formula that we construct, the relations will be
represented by relation variables $X,X_V,X_E,X_{\bar a},X_{\bar b},X_\eta$ of
appropriate types. $R^{(\ell)}$ will be the value of $X$ after the $\ell$th
  iteration of the main fixed-point iteration (and similarly for the other variables).

In the base step of our induction, we define $R^{(1)}$ to consist of all
tuples $(\bar s,u)\in F(G) \times V(G)$ such that $\{u\}\in\comp(G,\bar s)$. Then
defining the relations $R^{(1)}_V,\ldots,R^{(1)}_\eta$, that is, the anchored copy
of $\{u\}$, is easy, because the anchored copy only has a constant number of
elements, namely vertices corresponding to $u$ and to the vertices from the split pair $(\bar a,\bar b)$ of the flip extension $\bar s$.

So let us turn to the inductive step. We have already defined
relations $R^{(\ell)},R^{(\ell)}_V,\ldots,R^{(\ell)}_\eta$. We look at a flip
extension $\bar s$ and a vertex $u\in V(G)$ such that
$(\bar s,u)\not\in R^{(\ell)}$, that is, we have not yet defined an
anchored copy of the component $C\in\comp(G,\bar s)$ that contains
$u$. For all $\bar s_1,\bar s_2\in F(G)$, we do the following. We let
$\mathcal D_i$ be the set of all $D\in\comp(G,\bar s_i)$ such that
$(\bar s_i,v)\in R^{(\ell)}$ for some $v\in D$. By induction, this means
that actually $(\bar s_i,v)\in R^{(\ell)}$ for all $v\in D$ and that
we have already computed an anchored copy $D_\anc$ of $D$, which is
represented by the $(\bar s_i,v)$-entries of the relations 
in $R^{(\ell)}_V,\ldots,R^{(\ell)}_\eta$. We let $\mathcal D_i^\anc$ be
the set of all these anchored copies $D_\anc$. Now we apply the
$\FPC$-formulas of Lemma~\ref{lem:inductive-step} to $(G,\bar s,\bar
s_1,\bar s_2,C,\mathcal D_1,\mathcal D_1^\anc, \mathcal D_2,\mathcal
D_2^\anc)$. We obtain a structure $C_{\bar s_1,\bar s_2}=(P_V,P_E,\bar
p_a,\bar p_b,P_\eta)$. Note that $C_{\bar s_1,\bar s_2}$ is not
necessarily an $\bar s$-anchored copy of $C$, because we do not know whether
there are sets $X,X_1,X_2$ such that $\bar s,\bar s_1,\bar s_2$
satisfy the assumptions of Lemma~\ref{lem:inductive-step}. However,
using Lemma~\ref{lem:test-anchored}, we can check if $C_{\bar s_1,\bar s_2}$
is an anchored copy, regardless of whether the assumptions of
Lemma~\ref{lem:inductive-step} are satisfied. If $C_{\bar s_1,\bar s_2}$
is an $\bar s$-anchored copy of $C$, we call $(\bar s_1,\bar s_2)$ \emph{good} for $\bar s$.
 
 If there are $\bar s_1,\bar s_2\in F(G)$ that are good for $\bar s$,
 we add $(\bar s,u)$ to $R^{(\ell+1)}$. We let $C_\anc =(P_V,P_E,\bar
 p_a,\bar p_b,P_\eta)$ be the lexicographically smallest of all
 structures $C_{\bar s_1,\bar s_2}$, and we add
 \begin{itemize}
  \item all tuples $(\bar s,u,p)$ for $p\in P_V$ to $R_V^{(\ell+1)}$;
  \item all tuples $(\bar s,u,p_1,p_2)$ for $(p_1,p_2)\in P_E$ to $R_E^{(\ell+1)}$;
  \item the tuple $(\bar s,u,\bar p_a)$  to $R_{\bar a}^{(\ell+1)}$; 
  \item the tuple $(\bar s,u,\bar p_b)$  to $R_{\bar b}^{(\ell+1)}$; 
  \item the tuple $(\bar s,u,\bar p)$ for $\bar p \in P_\eta$ to $R_\eta^{(\ell+1)}$.
 \end{itemize}
 This completes the description of the inductive construction.

 It is not yet clear what the inductive process actually achieves, because
 it is not clear that in the inductive step we find any good tuples.
 To prove that the inductive process will eventually produce an
 ordered copy of $G$, we take a branch decomposition $(T,\gamma)$ of
 $G$ of width $k$. We prove that for every node
 $t\in V(T)$, every flip extension $\bar s$ for the set $X=\gamma(t)$,
 and every component $C\in\comp(G,\bar s)$ such that $C \subseteq X$, there is an $\ell\ge 1$
 such that $(\bar s,u)\in R^{(\ell)}$ for all $u\in C$. Indeed, we can
 choose $\ell$ to be $1$ plus the depth of $t$ in the tree, that is, the
 maximum length of a (directed) path from $t$ to a leaf. 
 
 The proof is by induction on $T$. The base step is trivial, because
 for leaves $t$,  we have $|\gamma(t)|=1$. For the inductive step $\ell\to\ell+1$, let
 $t$ be a node of depth $\ell+1$ with children $t_1,t_2$, and let $X=\gamma(t)$ and
 $X_i=\gamma(t_i)$. Then $X_1\cup X_2=X$ and $X_1\cap
 X_2=\emptyset$. Let $(\bar a,\bar b)$ be an ordered split pair for $X$
 and $\bar s$ a flip extension of $(\bar a,\bar b)$. Then, by
 Lemma~\ref{la:find-nice-split-pairs}, there exist
 ordered split pairs $(\bar a_i,\bar b_i)$ for $X_i$ that are nice with respect to $(\bar a,\bar b)$.
 Moreover, by Lemma~\ref{la:find-flip-extension} we can choose a flip extension $\bar s_i$
 of $(\bar a_i,\bar b_i)$ such that for every component
 $D\in\comp(G,\bar s_i)$, either $D\subseteq X_i$ or $D\cap
 X_i=\emptyset$. Let $\mathcal D_i$ be the set of all components
 $D\in\comp(G,\bar s_i)$ such that $(\bar s_i,v)\in R_i^{(\ell)}$ for
 all $v\in D$. By the induction hypothesis, for every component
 $D\in\comp(G,\bar s_i)$ with $D\subseteq X_i$ we have $D\in\mathcal
 D_i$. This implies, by Lemma~\ref{lem:inductive-step}, that the pair $(\bar
 s_1,\bar s_2)$ is good for $\bar s$. But then $(\bar s,u)\in
 R^{(\ell+1)}$ for every $u\in C$.
 
 There is a small problem at the root $r$ of $T$ because for
 $X=\gamma(r)=V(G)$ there is no split tuple. 
 However, we can apply the
 same construction as in the inductive step with $\bar s$ being
 the empty tuple. The problem is only a syntactic one:
 inductions formalised in $\FPC$ can only define relations of a fixed
 type, so we cannot directly replace the $2k+2^{2k}$ tuple $\bar s$ by
 the empty tuple. To resolve this, we carry out the last step of the inductive
 process separately adapting the types accordingly.
 In the end, we obtain the desired ordered copy of $G$.
\end{proof}

\section{Conclusions}

In this paper we considered the isomorphism and canonisation problem
for graphs of bounded rank width.  The first main result is that the
Weisfeiler-Leman dimension of graphs of rank width at most $k$ is at
most $3k+4$, that is, the $(3k+4)$-dimensional Weisfeiler-Leman algorithm
identifies all graphs of rank width at most $k$. This implies that
isomorphism testing and canonisation for graphs of rank width at most
$k$ can be done in time $n^{O(k)}$.

The second main result is that fixed-point logic with counting
captures polynomial time on the class of graphs of rank width at most $k$.

We remark that it is not difficult to obtain an $\Omega(k)$ lower
bound on the Weisfeiler-Leman dimension of graphs of rank width $k$.
Actually, combining a recent lower bound on the Weisfeiler-Leman dimension
for graphs of bounded tree width \cite{KieferN22} and Theorem \ref{thm:bound-rw-tw},
the Weisfeiler-Leman dimension of graphs of rank width at most $k$ is at least $\lfloor\frac{k}{2}\rfloor - 2$.
Thus our upper bound is asymptotically tight up to a factor of six.
Naturally, it would be nice to close or further narrow the gap between the upper and lower bound.

A more important question is whether isomorphism testing is also fixed-parameter tractable when parameterized by rank width.
We remark that fpt algorithms for isomorphism testing parameterized by tree width are known \cite{LokshtanovPPS17,GroheNSW20}.

An interesting open question on the logical side is whether rank
decompositions can be defined in monadic second order logic.  A
partial result for graphs of bounded linear clique width has been
obtained in \cite{BojanczykGP21}.  We believe the techniques
developed in this paper might also prove helpful for resolving the
general question.

\bibliographystyle{plainurl}
\bibliography{literature}

\appendix

\section{Canonisation from Weisfeiler-Leman}
\label{sec:app-canon-from-wl}

In this section we give a proof of Theorem \ref{thm:canon-from-wl}.
Towards this end, an intermediate theorem needs to be proven first.

Let $G$ be a graph.
The $k$-dimensional Weisfeiler-Leman algorithm \emph{determines orbits of $G$} if, for every graph $H$, every $v \in V(G)$ and every $w \in V(H)$ such that $\chi^{G,k}_{(\infty)}(v,\dots,v) = \chi^{H,k}_{(\infty)}(w,\dots,w)$,
there is an isomorphism $\varphi\colon G \cong H$ such that $\varphi(v) = w$.

\begin{theorem}
 \label{thm:orbits-from-wl}
 Let $\mathcal{C}$ be a class of graphs such that the $k$-dimensional Weisfeiler-Leman algorithm identifies all (coloured) graphs $G \in \mathcal{C}$.
 Then the $(k+1)$-dimensional Weisfeiler-Leman algorithm determines orbits of all (coloured) graphs $G \in \mathcal{C}$.
\end{theorem}

\begin{proof}
 Let $G \in \mathcal{C}$ and let $H$ be an arbitrary graph. Also let $v \in V(G)$ and $w \in V(H)$ such that $\chi^{G,k+1}_{(\infty)}(v,\dots,v) = \chi^{H,k+1}_{(\infty)}(w,\dots,w)$.
 Then $(G,\chi_G^{(v)}) \simeq_k (H,\chi_H^{(w)})$.
 Since the $k$-dimensional Weisfeiler-Leman algorithm identifies all graphs $G \in \mathcal{C}$ this implies that $(G,\chi_G^{(v)}) \cong (H,\chi_H^{(w)})$.
 So there is an isomorphism $\varphi\colon G \cong H$ such that $\varphi(v) = w$.
\end{proof}

\begin{theorem}[Theorem \ref{thm:canon-from-wl} restated]
 Let $\mathcal{C}$ be a graph class and suppose the $k$-dimensional Weisfeiler-Leman algorithm identifies all coloured graphs in $\mathcal{C}$.
 Then there is a graph canonisation for $\mathcal{C}$ that can be computed in time $O(n^{k+3}\log n)$.
\end{theorem}

\begin{algorithm}
 \caption{Canonisation Algorithm for graph class $\mathcal{C}$}
 \label{alg:canon-from-wl}
 
 \SetKwInOut{Input}{Input}
 \SetKwInOut{Output}{Output}
 \Input{Graph $G \in \mathcal{C}$}
 \Output{$\kappa(G)$}
 \BlankLine
 $n \coloneqq |V(G)|$\;
 $G_0 \coloneqq G$\;
 \For{$i = 1,\dots,n$}{
  compute $\chi_{G,i}(v) \coloneqq \chi^{G_{i-1},k+1}_{(\infty)}(v,\dots,v)$ for all $v \in V(G)$\;
  \tcc{if there is no unique minimum, the $\argmin$ operator picks an arbitrary element that minimises $\chi_{G,i}(v)$}
  $v_i \coloneqq \argmin_{v \in V(G) \setminus \{v_1,\dots,v_{i-1}\}} \chi_{G,i}(v)$\;
  $G_i \coloneqq (V(G),E(G),(\chi_{G,i})^{(v_i)})$\;
 }
 \Return $([n],\{ij \mid v_iv_j \in E(G)\},i \mapsto \chi(v_i))$\;
\end{algorithm}

\begin{proof}
 Let $\kappa\colon \mathcal{C} \rightarrow \mathcal{G}_\mathbb{N}$ be the function computed by Algorithm \ref{alg:canon-from-wl}.
 It is first argued that $\kappa$ canonises the graph class $\mathcal{C}$.
 Let $G \in \mathcal{C}$.
 Clearly, $\varphi\colon V(G) \rightarrow [n]\colon v_i \mapsto i$ is an isomorphism from $G$ to $\kappa(G)$.
 
 So let $H \in \mathcal{C}$ be a second graph such that $G \cong H$.
 Also let $v_1,\dots,v_n$ be the sequence of vertices computed by Algorithm \ref{alg:canon-from-wl} for the graph $G$ and let $w_1,\dots,w_n$ be the corresponding sequence for $H$.
 We prove by induction on $i \in \{0,\dots,n\}$ that there is an isomorphism $\varphi\colon G \cong H$ such that $\varphi(v_j) = w_j$ for all $j \leq i$.
 The base step $i = 0$ is exactly the assumption $G \cong H$.
 So let $i \geq 1$ and let $\varphi\colon G \cong H$ such that $\varphi(v_j) = w_j$ for all $j \leq i-1$.
 Then $(G,\chi_{G,i}) \cong (H,\chi_{H,i})$ and $\chi_{G,i}(v_i) = \chi_{H,i}(w_i)$.
 Since the $(k+1)$-dimensional Weisfeiler-Leman algorithm determines orbits for all graphs $G \in \mathcal{C}$ it follows that there is an isomorphism $\varphi\colon(G,\chi_{G,i}) \cong (H,\chi_{H,i})$ such that $\varphi(v_i) = w_i$.
 But this isomorphism has to map $v_j$ to $w_j$ for all $j \leq i$ since they have their own colour in the colouring $\chi_{G,i}$.
 
 By the induction principle, $\varphi\colon V(G) \rightarrow V(H)\colon v_i \mapsto w_i$ is an isomorphism from $G$ to $H$.
 Thus, $\kappa(G) = \kappa(H)$.
 
 The bound on the running time is immediately clear as the algorithm performs $n$ calls to the $(k+1)$-dimensional Weisfeiler-Leman algorithm,
 which runs in time $O(n^{k+2}\log n)$.
\end{proof}

\end{document}